\documentclass[12pt,pra,aps,amssymb,amsfonts,amsmath,tightenlines]{revtex4}
\usepackage{dsfont}
\usepackage{graphicx}
\usepackage{amssymb,amsfonts,amsthm}
\usepackage{color}
\usepackage{verbatim}
\usepackage[normalem]{ulem}
\usepackage{enumerate}
\usepackage{mathrsfs}
\usepackage{bm,float}
\newcommand{\diff}{\mathrm{d}}

\newtheorem{Proposition}{Proposition}
\newtheorem{Theorem}{Theorem}

\newtheorem{Lemma}{Lemma}
\newtheorem{Definition}{Definition}

\usepackage{nicefrac}
\newcommand{\Var}{\text{Var}}

\newcommand{\re}{\mbox{$\rm e$}}

\newcommand{\rd}{\mbox{$\rm d$}}

\begin{document}

\title{Pricing with Variance Gamma Information}

\author{ Lane~P.~Hughston$^1$ and  Leandro S\'anchez-Betancourt${}^2$}

\affiliation{
$^1$Department of Computing, Goldsmiths College, University of London,\\ New Cross, London SE14\,6NW, United Kingdom \\
$^2$Mathematical Institute, University of Oxford,  Oxford OX2 6GG, United Kingdom
}

\begin{abstract}
\noindent 
In the information-based pricing framework of Brody, Hughston \& Macrina, the market filtration  $\{ \mathcal F_t\}_{t\geq 0}$ is generated by an information process $\{ \xi_t\}_{t\geq0}$ defined in such a way that at some fixed time $T$ an $\mathcal F_T$-measurable random variable $X_T$ is ``revealed". A cash flow $H_T$ is taken to depend on the market factor $X_T$, and one considers the valuation of a financial asset that delivers $H_T$ at $T$. The value of the asset $S_t$  at any time $t\in[0,T)$ is the discounted conditional expectation of $H_T$ with respect to $\mathcal F_t$, where the expectation is under the risk neutral measure and the interest rate is constant. Then $S_{T^-} = H_T$, and $S_t = 0$ for $t\geq T$. In the general situation one has a countable number of cash flows, and each cash flow can depend on a vector of market factors, each associated with an information process. In the present work we introduce a new process, which we call the normalized variance-gamma bridge. We show that the normalized variance-gamma bridge and the associated gamma bridge are jointly Markovian. From these processes, together with the specification of a market factor $X_T$, we construct a so-called variance-gamma information process.  The filtration is then taken to be generated by the information process together with the gamma bridge. We show that the resulting extended information process has the Markov property and hence can be used to develop pricing models for a variety of different financial assets, several examples of which are discussed in detail.
\vspace{-0.2cm}
\\
\begin{center}
{\scriptsize {\bf Key words: Information-based asset pricing, L\'evy processes, gamma processes, \\ variance gamma processes, Brownian bridges, gamma bridges, 
nonlinear filtering.
} }
\end{center}
\end{abstract}

\maketitle
\section{Introduction}
\label{Introduction} 

\noindent
The theory of information-based asset pricing put forward by Brody, Hughston \& Macrina \cite{Macrina2006, BHM2007, BHM2008, BHM2008dam} is concerned with the determination of the price processes of financial assets from first principles. In particular, the market filtration is constructed explicitly, rather than simply assumed, as it is in traditional approaches.  The  simplest version of the model is as follows. We fix a probability space $\left(\Omega, \mathscr{F},\mathbb{P}\right)$. An asset delivers a single random cash flow $H_T$ at some specified time $T>0$, where time $0$ denotes the present. The cash flow is a function of a random variable $X_T$, which we can think of as a ``market factor'' that is in some sense revealed at time $T$. In the general situation there will be many factors and many cash flows, but for the present we assume that there is a single factor $X_T:\Omega\to \mathbb{R}$ such that the sole cash flow at time $T$ is given by $H_T=h(X_T)$ for some Borel function $h:\mathbb{R}\to \mathbb{R}^+$. For simplicity we assume that interest rates are constant and that $\mathbb{P}$ is the risk neutral measure. We require that $H_T$ should be integrable. Under these assumptions, the value of the asset at time $0$ is 

\begin{equation}
S_0=\re^{-r\,T}\,\mathbb{E}\left[h(X_T)\right],
\end{equation}
where $\mathbb{E}$ denotes expectation under $\mathbb{P}$ and $r$ is the short rate. Since the single ``dividend" is paid at time $T$, the value of the asset at any time $t\geq 0$ is of the form 

\begin{equation}
S_t=\re^{-r\,(T-t)}\,\mathds{1}_{\left\{t<T\right\}}\,\mathbb{E}\left[h(X_T)\,\big|\,\mathscr{F}_t \right],
\label{asset pricing equation}
\end{equation}
where $\{\mathscr{F}_t\}_{t\geq 0}$ is the market filtration. The task now is to model the filtration, and this will be done explicitly. 

In traditional financial modelling, the filtration is usually taken to be fixed in advance. For example, in the widely-applied Brownian-motion-driven model for financial markets, the filtration is generated by an $n$-dimensional Brownian motion. A detailed account of the Brownian framework can be found, for example, in Karatzas \& Shreve \cite{Karatzas Shreve}.  In the information-based approach, however, we do not assume the filtration to be given \textit{a priori}. Instead, the filtration is constructed in a way that specifically takes into account the structures of the information flows associated with the cash flows of  the various assets under consideration. 

In the case of a single asset generating a single cash flow, the idea is that the filtration should contain partial or ``noisy'' information about the market factor $X_T$, and hence the impending cash flow, in such a way that $X_T$ is $\mathscr{F}_T$-measurable. This can be achieved by allowing $\{\mathscr{F}_t\}$ to be generated by a so-called information process 
$\{\xi_t\}_{t\geq 0}$ with the property that for each $t$ such that $t\geq T$ the random variable $\xi_t$ is $\sigma\{{X_T}\}$-measurable. Then by constructing specific examples of c\'adl\`ag processes having this property, we are able to formulate a variety of specific models. The resulting models are  finely tuned to the structures of the assets that they represent, and therefore offer scope for a useful approach to financial risk management. 
In previous work on information-based asset pricing, where precise definitions can be found that expand upon the ideas summarized above, such models have been constructed using Brownian bridge information processes \cite{Macrina2006, BHM2007, BHM2008, Rutkowski Yu, BDF2009, BHM2010, BHM2011, FHM2012, HM2012, Menguturk 2013}, gamma bridge information processes \cite{BHM2008dam}, L\'evy random bridge information processes \cite{hoyle2010,HHM2012,HHM2015, Menguturk 2018, HMM2020}, and Markov bridge information processes \cite{MS2019}. In what follows we present a new model for the market filtration, based on the variance-gamma process. The idea is to create a two-parameter family of information processes associated with the random market factor $X_T$. One of the parameters is the information flow-rate $\sigma$. The other is an intrinsic parameter $m$ associated with the variance gamma process. In the limit as $m$ tends to infinity, the variance-gamma information process reduces to the type of Brownian bridge information process considered by Brody, Hughston \& Macrina \cite{Macrina2006, BHM2007, BHM2008}. 

The plan of the paper is as follows. In Section \ref{Gamma Subordinators} we recall properties of the gamma process, introducing the so-called scale parameter $\kappa >0$ and shape parameter $m>0$. A standard gamma subordinator is defined to be a gamma process with $\kappa = 1/m$. The mean at time $t$ of a standard gamma subordinator is $t$. In Theorem 1 we prove that an increase in the shape parameter $m$ results in a transfer of weight from the L\'evy measure of any interval $[c, d]$ in the space of jump size to the L\'evy measure of any interval $[a, b]$ such that $b-a = d-c$ and $c>a$. Thus, roughly speaking, an increase in $m$ results in an increase in the rate at which small jumps occur relative to the rate at which large jumps occur. This result concerning the interpretation of the shape parameter for a standard gamma subordinator is new as far as we are aware. 

In Section \ref{Normalized Variance-Gamma Bridge} we recall properties of the variance-gamma process and the gamma bridge, and in Definition \ref{def: Gamma} we introduce a new type of process, which we call a normalized variance-gamma bridge. This process plays an important role in the material that follows. In Lemmas \ref{Independence of Gamma and gamma} and \ref{Indep2} we work out various properties of the normalized variance-gamma bridge. Then in Theorem \ref{Markov_A_and_GammaBridge} we show that the normalized variance-gamma bridge and the associated gamma bridge are jointly Markov, a property that turns out to be crucial in our pricing theory. 
In Section \ref{Variance Gamma Information}, at Definition \ref{def: information process}, we introduce the so-called variance-gamma information process. The information process carries noisy information about the value of a market factor $X_T$ that will be revealed to the market at time $T$, where the noise is represented by the normalized variance-gamma bridge.  In equation \eqref{eq: identity information} we present a formula that relates the values of the information process at different times, and by use of that we establish in Theorem 3 that the information process and the associated gamma bridge are jointly Markov. 

In Section \ref{Asset Pricing}, we consider a market where the filtration is generated by a variance gamma information process along with the associated gamma bridge. In Lemma \ref{conditional density lemma} we work out a version of the Bayes formula in the form that we need for asset pricing in the present context. Then in  Theorem \ref{thm: value v_t formula} we present a general formula for the price process of a financial asset that at time $T$ pays a single dividend given by a function $h(X_T)$ of the market factor. In particular, the \textit{a priori} distribution of the market factor  can be quite arbitrary, specified by a measure $F_{X_T}(\rd x)$ on $\mathbb R$, and the only requirement being that $h(X_T)$ should be integrable. In Section \ref{Examples} we present a number of examples, based on various choices of the payoff function and the distribution for the market factor, the results being summarized in Propositions 1, 2, 3, and 4. We conclude with  comments on calibration, derivatives, and how one determines the trajectory of the information process from market prices.

\section{Gamma Subordinators}
\label{Gamma Subordinators}
\noindent We begin with some remarks about the gamma process. Let us as usual write $\mathbb R^+$ for the non-negative real numbers. Let $\kappa$ and $m$ be strictly positive constants. A continuous random variable $G: \Omega \to \mathbb R^+$ on a probability space $\left(\Omega,\,\mathscr{F},\,\mathbb{P}\right)$ will be said to have a gamma distribution with scale parameter $\kappa$ and shape parameter $m$ if

\begin{equation}
\mathbb{P}\left[G \in \diff x\right]    = \mathds{1}_{\left\{x>0\right\}}\,\frac{1}{\Gamma[m]}\,\kappa^{-m}\,x^{m-1}\,\re^{-x/\kappa}\,\diff x \, ,
\end{equation}
where 
\begin{equation}
\Gamma[a]=\int_0^{\infty} x^{a-1}\,\re^{-x}\,\diff x
\end{equation}
denotes the standard gamma function for $a>0$, and we recall the relation $\Gamma[a+1]=a\Gamma[a]$. A  calculation shows that $\mathbb{E}\left[G\right]=\kappa\,m$, and $\Var[G]=\kappa^2\,m$. There exists a two-parameter family of gamma processes of the form $\Gamma: \Omega \times \mathbb R^+ \to \mathbb R^+$ on $\left(\Omega,\,\mathscr{F},\,\mathbb{P}\right)$.  By a gamma process with scale $\kappa$ and shape $m$ we mean a L\'evy process $\{\Gamma_t\}_{t\geq 0}$ such that for each $t>0$ the random variable $\Gamma_t$ is gamma distributed with 

\begin{equation}
\mathbb{P}\left[\Gamma_t \in \diff x\right]    = \mathds{1}_{\left\{x>0\right\}}\,\frac{1}{\Gamma[m\,t]}\,\kappa^{-m\,t}\,x^{m\,t-1}\,\re^{-x/\kappa}\,\diff x \,.
\end{equation}
If we write $(a)_0=1$ and $(a)_k =a(a+1)(a+2) \cdots
(a+k-1)$ for the so-called Pochhammer symbol, we find that ${\mathbb E}
[\Gamma_t^n]= \kappa^n (mt)_n$. It follows that $\mathbb{E}[\Gamma_t]=\mu\,t$ and 
$\Var[\Gamma_t]=\nu^2\,t$, where $\mu=\kappa\,m$ and $\nu^2=\kappa^2\,m$, or equivalently $m=\mu^2/\nu^2$, and $\kappa=\nu^2/\mu$. 

The L\'evy exponent for such a process is given for $\alpha < 1$ by

\begin{equation}
\psi_{\Gamma}(\alpha)= \frac{1}{t} \log \mathbb E \left [\exp (\alpha \Gamma_t) \right] = -m\,\log \left(1-\kappa \alpha \right) ,
\end{equation}
and for the corresponding L\'evy measure we have

\begin{equation}
\nu_{\Gamma}(\diff x)=\mathds{1}_{\left\{x>0\right\}}\, m \,\frac{1}{x}\,\re^{-x/\kappa }\,\diff x \, .
\end{equation}
One can then check that the L\'evy-Khinchine relation

\begin{equation}
\psi_{\Gamma}(\alpha) = \int_{\mathbb R} \left(\re^{\alpha x} - 1 - \mathds{1}_{\left\{|x|<1\right\}}\, \alpha x \right) \, \nu_{\Gamma}(\diff x) + p \alpha
\end{equation}
holds for an appropriate choice of $p$  \cite[Lemma 1.7]{kyprianou2014fluctuations}. 

By a standard gamma subordinator we mean a gamma process $\{\gamma_t\}_{t\geq0}$ for which $\kappa=1/m$. This implies that $\mathbb{E}[\gamma_t]=t$ and $\Var[\gamma_t]=m^{-1}\,t$. The standard gamma subordinators thus constitute a one-parameter family of processes labelled by $m$. An interpretation of the parameter $m$ is given by the following: 

 \vspace{0.2cm}
\begin{Theorem}
Let $\{\gamma_t\}_{t\geq0}$ be a standard gamma subordinator with parameter $m$. Let $\nu_m[a,b]$ be the  L\'evy measure of the interval $[a,b]$ for $0<a<b$. Then for any interval $[c,d]$ such that $c>a$ and $d-c=b-a$ the ratio
\begin{equation}\label{eq: quotient Levy measures}
R_m(a,b\,;c,d) = \frac{\nu_m[a,b]}{\nu_m[c,d]}    
\end{equation}
is strictly greater than one and strictly increasing as a function of $m$.
\end{Theorem}
\begin{proof}
By the definition of a standard gamma subordinator we have  

\begin{equation}\label{eq: gamma levy measure a,b}
\nu_m[a,b]=\int_a^b m\,\frac{1}{x}\,\re^{-m\,x}\,\diff x\,.
\end{equation}
Let $\delta=c-a>0$ and note that the integrand in the right hand side of \eqref{eq: gamma levy measure a,b} is a decreasing function of the variable of integration. This allows one to conclude that 

\begin{equation}
\nu_m[a+\delta,b+\delta]=\int_{a+\delta}^{b+\delta} m\,\frac{1}{x}\,\re^{-m\,x}\,\diff x <  \int_a^b m\,\frac{1}{x}\,\re^{-m\,x}\,\diff x\,,
\end{equation}
from which it follows that
$0<\nu_m[c,d]<\nu_m[a,b]$ and hence $R_m(a,b\,;c,d) > 1$. 
To show that $R_m(a,b;c,d)$ is strictly increasing as a function of $m$ we observe that 

\begin{equation}
\nu_m[a,b]= m \int_a^{\infty} \frac{1}{x}\,\re^{-m\,x}\,\diff x - m \int_b^{\infty} \frac{1}{x}\,\re^{-m\,x}\,\diff x\, 
= m\,\left(E_1[m\,a]-E_1[m\,b]\right),
\end{equation}
where the so-called exponential integral function $E_1(z)$ is defined for $z>0$ by

\begin{equation}
E_1(z)=\int_z^{\infty} \frac{\re^{-x}}{x}\diff x\,.
\end{equation}

\noindent See reference  \cite{AS1970}, Section 5.1.1, for properties of the exponential integral. Next, we compute the derivative of $R_m(a,b\,;c,d)$, which gives

\begin{align}\label{derivative of quotient}
\frac{\partial}{\partial m}R_m(a,b\,;c,d)=\frac{1}{m\,\left(E_1[m\,c]-E_1[m\,d]\right)}\,\re^{-m\,a}\,\left(1-\re^{-m\,\Delta}\right)\,\left(R_m(a,b\,;c,d)-\re^{m(c-a)}\right),
\end{align}
where

\begin{equation}
\Delta=d-c=b-a \,.
\label{Delta}
\end{equation}
We note that 

\begin{equation}
\frac{1}{m\,\left(E_1[m\,c]-E_1[m\,d]\right)}\,\re^{-m\,a}\,\left(1-\re^{-m\,\Delta}\right) \,>\, 0\,,
\end{equation}
which shows that the sign of the derivative in \eqref{derivative of quotient} is strictly positive if and only if

\begin{equation}
R_m(a,b\,;c,d) \,>\, \re^{m(c-a)}.
\label{R inequality}
\end{equation}
But clearly

\begin{equation}
\int_0^{\Delta\,m}\frac{\re^{-u}}{u+a\,m}\,\diff u \,>\, \int_0^{\Delta\,m}\frac{\re^{-u}}{u+c\,m}\,\diff u\,
\end{equation}
for $c>a$, which after a change of integration variables and use of \eqref{Delta} implies 

\begin{equation}
\re^{m\,a}\,\int_{a\,m}^{b\,m}\frac{\re^{-x}}{x}\,\diff x \,>\, \re^{m\,c}\,\int_{c\,m}^{d\,m}\frac{\re^{-x}}{x}\,\diff x ,
\end{equation}
which is equivalent to \eqref{R inequality}, and that completes the proof.
\end{proof}
\noindent We see therefore that the effect of an increase in the value of $m$ is to transfer weight from the L\'evy measure of any jump-size interval $[c,d] \subset \mathbb{R}^+$ to any possibly-overlapping smaller-jump-size  interval $[a,b] \subset \mathbb{R}^+$ of the same length. The L\'evy measure of such an interval is the rate of arrival of jumps for which the jump size lies in that interval.

\section{Normalized Variance-Gamma Bridge}
\label{Normalized Variance-Gamma Bridge} 
\noindent Let us fix a standard Brownian motion $\{W_t\}_{t\geq 0}$ on $\left(\Omega,\,\mathscr{F},\,\mathbb{P}\right)$ and an independent standard gamma subordinator $\{\gamma_t\}_{t\geq0}$ with parameter $m$. By a standard variance-gamma process with parameter $m$ we mean a time-changed Brownian motion $\{V_t\}_{t\geq0}$ of the form

\begin{equation}
V_t=W_{\gamma_t}\,.
\end{equation}
It is straightforward to check that $\{V_t\}$ is itself a L\'evy process, with L\'evy exponent 

\begin{equation}
\psi_V(\alpha)=-m\,\log \left(1-\frac{\alpha^2}{2\,m}\right)\,.
\end{equation}
Properties of the variance-gamma process, and financial models based on it, have been investigated extensively in \cite{Madan 1990, Madan Milne 1991, Madan Carr Chang 1998, Carr Geman Madan Yor 2002} and many other works.

The other object we require going forward is the gamma bridge \cite{BHM2008dam, Emery Yor 2004, Yor 2007}. 
Let $\{\gamma_t\}$ be a standard gamma subordinator with parameter $m$. For fixed $T>0$ the process $\{\gamma_{tT}\}_{t\geq 0}$ defined by

\begin{equation}\label{eq: gamma bridge}
    \gamma_{tT}=\frac{\gamma_t}{\gamma_T}
\end{equation}
for $0 \leq t \leq T$ and $\gamma_{tT} = 1$ for $t>T$ will be called a standard gamma bridge, with parameter $m$, over the interval $[0,T]$. One can check that for $0< t< T$ the random variable $\gamma_{tT}$ has a beta
distribution  \cite[pp.~6-9]{BHM2008dam}. In particular, one finds that its density is given by

\begin{eqnarray}
\mathbb{P}\left[\gamma_{tT} \in \diff y\right]  ={\mathds 1}_{\{0<y<1\}} \,
\frac{\,y^{mt-1}(1-y)^{m(T-t)-1}\,}{{B}[mt,m(T-t)]} \,\diff y \, ,
\end{eqnarray}
where

\begin{eqnarray}
{B}[a,b] = \frac{\,\Gamma[a] \, \Gamma[b]\,}{\Gamma[a+b]} \, .
\end{eqnarray}
It follows then by use of the
integral formula

\begin{eqnarray}
{B}[a,b]=\int^1_0\,y^{a-1}(1-y)^{b-1}\rd y 
\end{eqnarray}
that for all $n \in \mathbb N$ we have

\begin{eqnarray}
{\mathbb E}\left[\gamma_{tT}^{n}\right] = \frac{{B}
[mt+n,m(T-t)]}{{B}[mt,m(T-t)]} ,
\end{eqnarray}
and hence

\begin{eqnarray}
{\mathbb E}\left[\gamma_{tT}^{n}\right]  = \frac{(mt)_n}{(mT)_n} \, .
\end{eqnarray}
Accordingly, one has 

\begin{eqnarray}\label{eq: first and second moment of bridge}
{\mathbb E}[\gamma_{tT}]
=t/T\,,\quad {\mathbb E} [\gamma_{tT}^2]=t(mt+1)/T(mT+1)
\end{eqnarray}
and therefore

\begin{equation}
\textrm{Var}[\gamma_{tT}]= \frac{t(T-t)} {T^2(1+mT)} \,.
\end{equation}
One observes, in particular, that the
expectation of $\gamma_{tT}$ does not depend on $m$,
whereas the variance of $\gamma_{tT}$ decreases as $m$ increases.

 \vspace{0.2cm}

\begin{Definition}\label{def: Gamma}
For fixed $T>0$,  the process $\left\{\Gamma_{tT}\right\}_{t \geq0}$ defined by
\begin{equation}\label{eq: VG bridge}
\Gamma_{tT}={{\gamma_T}^{-\frac{1}{2}}}\,\left(W_{\gamma_t}-\gamma_{tT}\,W_{\gamma_T}\right)
\end{equation}
for $0 \leq t \leq T$ and $\Gamma_{tT} = 0$ for $t>T$ will be called a \textit{normalized variance gamma bridge}.
\end{Definition}

We proceed to work out various properties of this process. We observe that $\Gamma_{tT}$ is conditionally Gaussian, from which it follows that $\mathbb{E}\left[\Gamma_{tT} \mid \gamma_t,\, \gamma_T\right]=0$ and $\mathbb{E}\left[\Gamma^2_{tT} \mid \gamma_t,\, \gamma_T\right]=\gamma_{tT} \left(1-\gamma_{tT}\right)$. Therefore $\mathbb E[\Gamma_{tT}] = 0$  and $\mathbb E[\Gamma^2_{tT}]=\mathbb E[\gamma_{tT}]-\mathbb E[\gamma^2_{tT}]$\,; and thus by use of \eqref{eq: first and second moment of bridge} we have 

\begin{equation}
{\rm Var}\,[\Gamma_{tT}] =  \frac{\,m t \, (T-t)\,} {T\,(1 + mT)} \,.
\end{equation}
Now,  recall \cite{Yor 2007,Emery Yor 2004} that the gamma process and the associated gamma bridge have the following fundamental independence property. Define 

\begin{equation}
\mathscr{G}^*_t=\sigma\left\{\gamma_s/\gamma_t,\,\,s\in [0,t] \right\}\,,\quad \mathscr{G}^+_t=\sigma\left\{\gamma_u,\,\,u\in [t,\infty) \right\}\,.
\end{equation}
Then, for every $t\geq 0$ it holds that $\mathscr{G}^*_t$ and $\mathscr{G}^+_t$ are independent. In particular  $\gamma_{st}$ and $\gamma_{u}$ are independent for $0\leq s\leq t\leq u$ and $t>0$. It also holds that $\gamma_{st}$ and $\gamma_{uv}$ are independent for $0 \leq s\leq t\leq u \leq v$ and $t>0$. Furthermore, we have:
 \vspace{0.2cm}
 
\begin{Lemma}
If $0\leq s\leq t\leq u$ and $t>0$ then  $\Gamma_{st}$ and ${\gamma_{u}}$ are independent. 
\label{Independence of Gamma and gamma}
\end{Lemma}
\begin{proof}
We recall that if a random variable $X$ is normally distributed with mean $\mu$ and variance $\nu^2$ then

\begin{equation}
\mathbb{P}\left[X<x\right]=N\left(\frac{x-\mu}{\nu}\right),
\end{equation}
where $N:\mathbb{R}\to (0,1)$ is defined by

\begin{equation}
N(x)= \frac{1}{\sqrt{2 \pi} } \int_{-\infty}^{x} \exp \left(-\frac{1}{2}\,y^2\right)\,\diff y\,.
\label{normal distribution function}
\end{equation}
Since $\Gamma_{tT}$ is conditionally Gaussian, by use of the tower property we find that

\begin{align}
F_{\,\Gamma_{st},\,\gamma_{u}}(x,y)&=\mathbb{E}\left[\mathds{1}_{\{\Gamma_{st}\leq x\}}\,\mathds{1}_{\{\gamma_{u}\leq y\}}\right] \nonumber\\
&=\mathbb{E}\left[ \mathbb{E}\left[ \mathds{1}_{\{\Gamma_{st}\leq x\}}\,\mathds{1}_{\{\gamma_{u}\leq y\}}\,\middle|\,\gamma_s\,,\gamma_t\,,\gamma_{u} \right]\right]\nonumber\\
&=\mathbb{E}\left[\mathds{1}_{\{\gamma_{u}\leq y\}}\, \mathbb{E}\left[ \mathds{1}_{\{\Gamma_{st}\leq x\}}\,\middle|\,\gamma_s\,,\gamma_t\,,\gamma_{u} \right]\right]\nonumber\\
&=\mathbb{E}\left[\mathds{1}_{\{\gamma_{u}\leq y\}}\,N\left({x}\,{\left(\gamma_{st}\,\left({1-\gamma_{st}}\right)\right)^{-\frac{1}{2}}}\right) \right]\nonumber\\
&=\mathbb{E}\left[\mathds{1}_{\{\gamma_{u}\leq y\}}\right]\mathbb{E}\left[N\left({x}\,{\left(\gamma_{st}\,\left({1-\gamma_{st}}\right)\right)^{-\frac{1}{2}}}\right) \right],
\end{align}
where the last line follows from the independence of  $\gamma_{st}$ and $\gamma_{u}$.
\end{proof}
\vspace{0.1cm}
\noindent By a straightforward extension of the argument we deduce that if $0\leq s\leq t\leq u \leq v$ and $t>0$ then $\Gamma_{st}$ and ${\gamma_{uv}}$ are independent. Further, we have:

\vspace{0.1cm}
\begin{Lemma} \label{Indep2}
If  $\,0\leq s\leq t\leq u \leq v$ and $t>0$ then $\Gamma_{st}$ and $\Gamma_{uv}$ are independent.
\end{Lemma}

\noindent {\em Proof.} We recall that the Brownian bridge $\{\beta_{tT}\}_{0\leq t\leq T}$ defined by

\begin{equation}
\beta_{tT}=W_t-\frac{t}{T}\,W_T
\end{equation}
for $0 \leq t \leq T$ and $\beta_{tT} = 0$ for $t>T$ is Gaussian with $\mathbb{E}\left[\beta_{tT}\right]=0$, 
$\text{Var}\left[\beta_{tT}\right]=t\,(T-t)/T$, and $\text{Cov}\left[\beta_{sT},\,\beta_{tT}\right]=s(T-t)/T$ for $0\leq s\leq t \leq T$. Using the tower property we find that

\begin{align}
F_{\Gamma_{st},\,\Gamma_{uv}}(x,y)&=\mathbb{E}\left[\mathds{1}_{\{\Gamma_{st}\leq x\}}\,\mathds{1}_{\{\Gamma_{uv}\leq y\}}\right]\nonumber\\
&=\mathbb{E}\left[ \mathbb{E}\left[ \mathds{1}_{\{\Gamma_{st}\leq x\}}\,\mathds{1}_{\{\Gamma_{uv}\leq y\}}\,\middle|\,\gamma_s\,,\gamma_t\,,\gamma_{u}\,,\gamma_{v}\right]\right]\nonumber\\
&=\mathbb{E}\left[ \mathbb{E}\left[ \mathds{1}_{\{\Gamma_{st}\leq x\}}\,\middle|\,\gamma_s\,,\gamma_t\,,\gamma_{u}\,,\gamma_{v}\right]\,\mathbb{E}\left[ \mathds{1}_{\{\Gamma_{uv}\leq y\}}\,\middle|\,\gamma_s\,,\gamma_t\,,\gamma_{u}\,,\gamma_{v}\right]\right]\nonumber\\
&=\mathbb{E}\left[ N\left({x}\,{\left(\left({1-\gamma_{st}}\right)\left(\gamma_{st}\right)\right)^{-\frac{1}{2}}}\right)\right] \,\mathbb{E}\left[ N\left({y}\,{\left(\left({1-\gamma_{uv}}\right)\left(\gamma_{uv}\right)\right)^{-\frac{1}{2}}}\right)\right] ,
\end{align}
where in the final step we use \eqref{eq: VG bridge} along with properties of the Brownian bridge.

 \hfill $\Box$\\
 
\vspace{0.2cm}

\noindent A straightforward calculation shows that if $0\leq s\leq t\leq u$ and $t>0$ then

\begin{equation}\label{eq: Gamma identity}
    \Gamma_{su}=\left({\gamma_{tu}}\right)^{\frac{1}{2}}\,\Gamma_{st}+{\gamma_{st}}\,\Gamma_{tu}\,.
\end{equation}

\noindent With this result at hand we obtain the following:

\vspace{0.2cm}
\begin{Theorem}\label{Markov_A_and_GammaBridge}
The processes $\{\Gamma_{tT}\}_{0\leq t \leq T}$ and $\{\gamma_{tT}\}_{0\leq t \leq T}$ are jointly Markov.
\end{Theorem}

\begin{proof}
To establish the Markov property it suffices to show that for any bounded measurable function  $\phi:\mathbb{R}\times\mathbb{R}\to\mathbb{R}$, any $n\in\mathbb{N}$, and any  $0\leq t_n\leq t_{n-1}\leq \,\dots\,\leq t_1\leq t \leq T$, we have

\begin{align}
&\mathbb{E}\left[\phi(\Gamma_{tT},\gamma_{tT})\,\middle|\,\Gamma_{t_1 T},\,\gamma_{t_1 T},\,\Gamma_{t_2 T},\,\gamma_{t_2 T},\,\dots\,,\Gamma_{t_n T},\,\gamma_{t_n T}\right]\nonumber\\
&\quad\quad=\mathbb{E}\left[\phi(\Gamma_{tT},\gamma_{tT})\,\middle|\,\Gamma_{t_1 T},\,\gamma_{t_1 T}\right]\,.
\end{align}

\noindent We present the proof for $n=2$. Thus we need to show that

\begin{align}
&\mathbb{E}\left[\phi(\Gamma_{tT},\gamma_{tT})\,\middle|\,\Gamma_{t_1 T},\,\gamma_{t_1 T},\,\Gamma_{t_2 T},\,\gamma_{t_2 T}\right]\nonumber\\
&\quad\quad=\mathbb{E}\left[\phi(\Gamma_{tT},\gamma_{tT})\,\middle|\,\Gamma_{t_1 T},\,\gamma_{t_1 T}\right]\,.
\end{align}

\noindent As a consequence of \eqref{eq: Gamma identity} we have

\begin{align}
&\mathbb{E}\left[\phi(\Gamma_{tT},\gamma_{tT})\,\middle|\,\Gamma_{t_1 T},\,\gamma_{t_1 T},\,\Gamma_{t_2 T},\,\gamma_{t_2 T}\right]\nonumber\\
&\quad\quad=\mathbb{E}\left[\phi(\Gamma_{tT},\gamma_{tT})\,\middle|\,\Gamma_{t_1 T},\,\gamma_{t_1 T},\,\Gamma_{t_2 t_1},\,\gamma_{t_2 t_1}\right]\,.
\end{align}

\noindent Therefore, it suffices to show that

\begin{align}\label{eq: Markov n=2 VGInfo}
&\mathbb{E}\left[\phi(\Gamma_{tT},\gamma_{tT})\,\middle|\,\Gamma_{t_1 T},\,\gamma_{t_1 T},\,\Gamma_{t_2 t_1},\,\gamma_{t_2 t_1}\right]\nonumber\\
&\quad\quad=\mathbb{E}\left[\phi(\Gamma_{tT},\gamma_{tT})\,\middle|\,\Gamma_{t_1 T},\,\gamma_{t_1 T}\right]\,.
\end{align}

\noindent Let us write 

\begin{equation}
f_{\,\Gamma_{tT},\,\gamma_{tT},\,{\Gamma_{t_1 T}},\,\gamma_{t_1 T},\,\Gamma_{t_2 t_1},\,\gamma_{t_2 t_1}}(x,\,y,\,a,\,b,\,c,\,d)
\end{equation}
for the joint density of ${\Gamma_{tT},\,\gamma_{tT},\,{\Gamma_{t_1 T}},\,\gamma_{t_1 T},\,\Gamma_{t_2 t_1},\,\gamma_{t_2 t_1}}$. Then for the conditional density of $\Gamma_{tT}$ and 
$\gamma_{tT}$ given ${{\Gamma_{t_1 T}}=a,\,\gamma_{t_1 T}=b,\,\Gamma_{t_2 t_1}=c,\,\gamma_{t_2 t_1}=d}$ we have

\begin{equation}
g_{\,\Gamma_{tT},\,\gamma_{tT}}(x,\,y,\,a,\,b,\,c,\,d)=\frac{f_{\,\Gamma_{tT},\,\gamma_{tT},\,{\Gamma_{t_1 T}},\,\gamma_{t_1 T},\,\Gamma_{t_2 t_1},\,\gamma_{t_2 t_1}}(x,\,y,\,a,\,b,\,c,\,d)}{f_{\,{\Gamma_{t_1 T}},\,\gamma_{t_1 T},\,\Gamma_{t_2 t_1},\,\gamma_{t_2 t_1}}(a,\,b,\,c,\,d)}\,.
\end{equation}

\noindent Thus,

\begin{align}
&\mathbb{E}\left[\phi(\Gamma_{tT},\gamma_{tT})\,\middle|\,\Gamma_{t_1 T},\,\gamma_{t_1 T},\,\Gamma_{t_2 t_1},\,\gamma_{t_2 t_1}\right]\nonumber\\
&\quad \quad= \int_{\mathbb{R}}\int_{\mathbb{R}} \phi(x,y)\, g_{\,\Gamma_{tT},\,\gamma_{tT}}(x,\,y,\,\Gamma_{t_1 T},\,\gamma_{t_1 T},\,\Gamma_{t_2 t_1},\,\gamma_{t_2 t_1})\,\diff x\,\diff y \,.
\end{align}

\noindent Similarly,

\begin{align}
&\mathbb{E}\left[\phi(\Gamma_{tT},\gamma_{tT})\,\middle|\,\Gamma_{t_1 T},\,\gamma_{t_1 T}\right]\nonumber\\
&\quad \quad= \int_{\mathbb{R}}\int_{\mathbb{R}} \phi(x,y)\, g_{\,\Gamma_{tT},\,\gamma_{tT}}(x,\,y,\,\Gamma_{t_1 T},\,\gamma_{t_1 T}) \, \diff x\,\diff y \,,
\end{align}
where for the conditional density of  $\Gamma_{tT}$ and 
$\gamma_{tT}$ given ${\Gamma_{t_1 T}}=a,\,\gamma_{t_1 T}=b$ we have  

\begin{equation}
g_{\,\Gamma_{tT},\,\gamma_{tT}}(x,\,y,\,a,\,b)=\frac{f_{\,\Gamma_{tT},\,\gamma_{tT},\,{\Gamma_{t_1 T}},\,\gamma_{t_1 T}}(x,\,y,\,a,\,b)}{f_{\,{\Gamma_{t_1 T}},\,\gamma_{t_1 T}}(a,\,b)}\,.
\end{equation}

Note that the conditional probability densities that we introduce  in formulae such as those above are ``regular" conditional densities \cite[p.~91]{williams1991}. We shall show  that 

\begin{equation}
g_{\,\Gamma_{tT},\,\gamma_{tT}}(x,\,y,\,\Gamma_{t_1 T},\,\gamma_{t_1 T},\,\Gamma_{t_2 t_1},\,\gamma_{t_2 t_1})=g_{\,\Gamma_{tT},\,\gamma_{tT}}(x,\,y,\,\Gamma_{t_1 T},\,\gamma_{t_1 T})\,.
\end{equation}

\noindent Writing 

\begin{align}
&F_{\,\Gamma_{tT},\,\gamma_{tT},\,{\Gamma_{t_1 T}},\,\gamma_{t_1 T},\,\Gamma_{t_2 t_1},\,\gamma_{t_2 t_1}}(x,\,y,\,a,\,b,\,c,\,d)\nonumber\\
&\quad \quad= \mathbb{E}\left[\mathds{1}_{\left\{\Gamma_{tT}<x\right\}} \mathds{1}_{\left\{ \gamma_{tT}<y\right\}}\,\mathds{1}_{\left\{\Gamma_{t_1 T}<a \right\}}\mathds{1}_{\left\{ \gamma_{t_1 T}<b\right\}}\,\mathds{1}_{\left\{\Gamma_{t_2 t_1}<c\right\}}\mathds{1}_{\left\{ \gamma_{t_2 t_1}<d\right\}}\right]
\end{align}
for the joint distribution function, we see that 

\begin{align}
&F_{\,\Gamma_{tT},\,\gamma_{tT},\,{\Gamma_{t_1 T}},\,\gamma_{t_1 T},\,\Gamma_{t_2 t_1},\,\gamma_{t_2 t_1}}(x,\,y,\,a,\,b,\,c,\,d)\nonumber\\
& =\mathbb{E}\left[\mathds{1}_{\left\{\Gamma_{tT}<x\right\}} \mathds{1}_{\left\{ \gamma_{tT}<y\right\}}\,\mathds{1}_{\left\{\Gamma_{t_1 T}<a \right\}}\mathds{1}_{\left\{ \gamma_{t_1 T}<b\right\}}\,\mathds{1}_{\left\{\Gamma_{t_2 t_1}<c\right\}}\mathds{1}_{\left\{ \gamma_{t_2 t_1}<d\right\}}\right]\nonumber\\
&=\mathbb{E}\left[\mathbb{E}\left[\mathds{1}_{\left\{\Gamma_{tT}<x\right\}}\, \mathds{1}_{\left\{ \gamma_{tT}<y\right\}}\,\mathds{1}_{\left\{\Gamma_{t_1 T}<a \right\}}\,\mathds{1}_{\left\{ \gamma_{t_1 T}<b\right\}}\,\mathds{1}_{\left\{\Gamma_{t_2 t_1}<c\right\}}\mathds{1}_{\left\{ \gamma_{t_2 t_1}<d\right\}}\,\middle|\,\,\gamma_{t_2},\,\gamma_{t_1},\,\gamma_{t},\,\gamma_T\right]\right]\nonumber\\
&=\mathbb{E}\left[\mathds{1}_{\left\{ \gamma_{tT}<y\right\}}\,\mathds{1}_{\left\{ \gamma_{t_1 T}<b\right\}}\,\mathds{1}_{\left\{ \gamma_{t_2 t_1}<d\right\}}\,\mathbb{E}\left[\mathds{1}_{\left\{\Gamma_{tT}<x\right\}} \,\mathds{1}_{\left\{\Gamma_{t_1 T}<a \right\}}\,\mathds{1}_{\left\{\Gamma_{t_2 t_1}<c\right\}}\,\middle|\,\,\gamma_{t_2},\,\gamma_{t_1},\,\gamma_{t},\,\gamma_T\right]\right]\nonumber\\
&=\mathbb{E}\Bigg[\mathbb{E}\left[\mathds{1}_{\left\{\Gamma_{tT}<x\right\}}\, \mathds{1}_{\left\{ \gamma_{tT}<y\right\}}\,\mathds{1}_{\left\{\Gamma_{t_1 T}<a \right\}}\,\mathds{1}_{\left\{ \gamma_{t_1 T}<b\right\}}\,\middle|\,\,\gamma_{t_1},\,\gamma_{t},\,\gamma_T\right] \nonumber\\
&\hspace{4.0cm} \times N\left(\frac{c}{\sqrt{\left(1-\gamma_{t_2 t_1}\right)\left(\gamma_{t_2 t_1}\right)}}\right)\,\mathds{1}_{\left\{\gamma_{t_2 t_1}<d\right\}}\Bigg] ,
\end{align}
where the last step follows as a consequence of Lemma \ref{Indep2}.
Thus we have

\begin{align}
&F_{\,\Gamma_{tT},\,\gamma_{tT},\,{\Gamma_{t_1 T}},\,\gamma_{t_1 T},\,\Gamma_{t_2 t_1},\,\gamma_{t_2 t_1}}(x,\,y,\,a,\,b,\,c,\,d)\nonumber\\
&=\mathbb{E}\left[\mathds{1}_{\left\{\Gamma_{tT}<x\right\}}\, \mathds{1}_{\left\{ \gamma_{tT}<y\right\}}\,\mathds{1}_{\left\{\Gamma_{t_1 T}<a \right\}}\,\mathds{1}_{\left\{ \gamma_{t_1 T}<b\right\}}\,N\left(\frac{c}{\sqrt{\left(1-\gamma_{t_2 t_1}\right)\left(\gamma_{t_2 t_1}\right)}}\right)\,\mathds{1}_{\left\{\gamma_{t_2 t_1}<d\right\}}\right]\nonumber\\
&=\mathbb{E}\left[\mathds{1}_{\left\{\Gamma_{tT}<x\right\}}\, \mathds{1}_{\left\{ \gamma_{tT}<y\right\}}\,\mathds{1}_{\left\{\Gamma_{t_1 T}<a \right\}}\,\mathds{1}_{\left\{ \gamma_{t_1 T}<b\right\}}\right]\,\mathbb{E}\left[N\left(\frac{c}{\sqrt{\left(1-\gamma_{t_2 t_1}\right)\left(\gamma_{t_2 t_1}\right)}}\right)\,\mathds{1}_{\left\{\gamma_{t_2 t_1}<d\right\}}\right]\label{eq: indep last step theorem 2}\nonumber\\
&=F_{\,\Gamma_{tT},\,\gamma_{tT},\,{\Gamma_{t_1 T}},\,\gamma_{t_1 T}}(x,\,y,\,a,\,b)\times F_{\,\Gamma_{t_2 t_1},\,\gamma_{t_2 t_1}}(c,\,d)\,,
\end{align}
where  the next to last step follows by virtue of the fact that $\Gamma_{st}$ and ${\gamma_{uv}}$ are independent for $0\leq s\leq t\leq u \leq v$ and $t>0$. Similarly,

\begin{align}
&F_{\,{\Gamma_{t_1 T}},\,\gamma_{t_1 T},\,\Gamma_{t_2 t_1},\,\gamma_{t_2 t_1}}(a,\,b,\,c,\,d)\nonumber\\
&=\mathbb{E}\left[\mathds{1}_{\left\{\Gamma_{t_1 T}<a\right\}}\mathds{1}_{\left\{ \gamma_{t_1 T}<b\right\}}\,\mathds{1}_{\left\{\Gamma_{t_2 t_1}<c \right\}}\mathds{1}_{\left\{ \gamma_{t_2 t_1}<d\right\}}\right]\nonumber\\
&=\mathbb{E}\left[\mathbb{E}\left[\mathds{1}_{\left\{\Gamma_{t_1 T}<a \right\}} \mathds{1}_{\left\{ \gamma_{t_1 T}<b\right\}}\,\mathds{1}_{\left\{\Gamma_{t_2 t_1}<c \right\}}\mathds{1}_{\left\{ \gamma_{t_2 t_1}<d\right\}}\,\middle|\,\,\gamma_{t_2},\,\gamma_{t_1},\,\gamma_T\right]\right]\nonumber\\
&=\mathbb{E}\left[\mathds{1}_{\left\{ \gamma_{t_1 T}<b\right\}}\,\mathds{1}_{\left\{ \gamma_{t_2 t_1}<d\right\}}\,\mathbb{E}\left[\mathds{1}_{\left\{\Gamma_{t_1 T}<a \right\}}\,\mathds{1}_{\left\{\Gamma_{t_2 t_1}<c\right\}}\,\middle|\,\,\gamma_{t_2},\,\gamma_{t_1},\,\gamma_T\right]\right] ,
\end{align}
and hence

\begin{align}
&F_{\,{\Gamma_{t_1 T}},\,\gamma_{t_1 T},\,\Gamma_{t_2 t_1},\,\gamma_{t_2 t_1}}(a,\,b,\,c,\,d)\nonumber\\
&=\mathbb{E}\left[N\left(\frac{a}{\sqrt{\left(1-\gamma_{t_1 T}\right)\left(\gamma_{t_1 T}\right)}}\right)\,\mathds{1}_{\left\{\gamma_{t_1 T}<b\right\}}\,N\left(\frac{c}{\sqrt{\left(1-\gamma_{t_2 t_1}\right)\left(\gamma_{t_2 t_1}\right)}}\right)\,\mathds{1}_{\left\{\gamma_{t_2 t_1}<d\right\}}\right]\nonumber\\
&=\mathbb{E}\left[N\left(\frac{a}{\sqrt{\left(1-\gamma_{t_1 T}\right)\left(\gamma_{t_1 T}\right)}}\right)\,\mathds{1}_{\left\{\gamma_{t_1 T}<b\right\}}\right]\,\mathbb{E}\left[N\left(\frac{c}{\sqrt{\left(1-\gamma_{t_2 t_1}\right)\left(\gamma_{t_2 t_1}\right)}}\right)\,\mathds{1}_{\left\{\gamma_{t_2 t_1}<d\right\}}\right]\nonumber\\
&=F_{\,{\Gamma_{t_1 T}},\,\gamma_{t_1 T}}(a,\,b)\times F_{\,\Gamma_{t_2 t_1},\,\gamma_{t_2 t_1}}(c,\,d)\,.
\end{align}

\noindent Thus we deduce that

\begin{align}
&f_{\,\Gamma_{tT},\,\gamma_{tT},\,{\Gamma_{t_1 T}},\,\gamma_{t_1 T},\,\Gamma_{t_2 t_1},\,\gamma_{t_2 t_1}}(x,\,y,\,a,\,b,\,c,\,d)\\
&\quad\quad\quad\quad\quad\quad=f_{\,\Gamma_{tT},\,\gamma_{tT},\,{\Gamma_{t_1 T}},\,\gamma_{t_1 T}}(x,\,y,\,a,\,b)\times f_{\,\Gamma_{t_2 t_1},\,\gamma_{t_2 t_1}}(c,\,d)\,,
\end{align}
and 

\begin{align}
f_{\,{\Gamma_{t_1 T}},\,\gamma_{t_1 T},\,\Gamma_{t_2 t_1},\,\gamma_{t_2 t_1}}(a,\,b,\,c,\,d)=f_{\,{\Gamma_{t_1 T}},\,\gamma_{t_1 T}}(a,\,b)\times f_{\Gamma_{t_2 t_1},\,\gamma_{t_2 t_1}}(c,\,d)\,,
\end{align}
and the theorem follows. 
\end{proof}

\section{Variance Gamma Information}
\label{Variance Gamma  Information}

\noindent Fix $T>0$ and let $\{\Gamma_{tT}\}$ be a normalized variance gamma bridge, as defined by \eqref{eq: VG bridge}. Let $\{\gamma_{tT}\}$ be the associated gamma bridge defined by \eqref{eq: gamma bridge}. Let $X_T$ be a random variable and assume that $X_T$, $\{\gamma_t\}_{t\geq 0}$ and $\{W_t\}_{t\geq 0}$ are independent. We are led to the following: 
\begin{Definition}\label{def: information process}
By a variance-gamma information process carrying the market factor $X_T$ we mean a process 
$\{\xi_t\}_{t\geq 0}$ that takes the form

\begin{equation}
\xi_t=\Gamma_{tT}+\sigma\,{\gamma_{tT}}\,X_T\,
\label{information process}
\end{equation}
for $0 \leq t \leq T$ and $\xi_t = \sigma X_T$ for $t > T$, where $\sigma$ is a positive constant. 
\end{Definition}

The market filtration is assumed to be the standard augmented filtration generated jointly by $\{\xi_t\}$ and $\{\gamma_{tT}\}$. A calculation shows that if $0\leq s\leq t\leq T$ and $t>0$ then

\begin{equation}\label{eq: identity information}
\xi_{s}=\Gamma_{st}\left(\gamma_{tT}\right)^{\frac{1}{2}}+\xi_{t}\,{\gamma_{st}}\,.    
\end{equation}

\noindent We are thus led to the following result required for the valuation of assets.

%
\vspace{0.1cm}
\begin{Theorem}\label{MarkovInformationProcess2}
The processes $\{\xi_t\}_{0 \leq t \leq T}$ and $\{\gamma_{tT}\}_{0 \leq t \leq T}$ are jointly Markov.
\end{Theorem}

\begin{proof}
It suffices to show that for any $n\in\mathbb{N}$ and $0<t_1<t_2<\dots<t_n$ we have 

\begin{equation}
\mathbb{E}\left[\phi(\xi_{t},\gamma_{tT})\,\middle|\,\xi_{t_1},\,\xi_{t_2},\,\dots\,,\xi_{t_n},\,\gamma_{t_1 T},\,\gamma_{t_2 T},\,\dots\,,\gamma_{t_n T}\right]
=\mathbb{E}\left[\phi(\xi_{t},\gamma_{tT})\,\middle|\,\xi_{t_1},\,\gamma_{t_1 T}\right] .
\end{equation}
We present the proof for $n=2$. Thus, we propose to show that 

\begin{equation}
\mathbb{E}\left[\phi(\xi_{t},\gamma_{tT})\,\middle|\,\xi_{t_1},\,\xi_{t_2},\,\gamma_{t_1 T},\,\gamma_{t_2 T}\right]
=\mathbb{E}\left[\phi(\xi_{t},\gamma_{tT})\,\middle|\,\xi_{t_1},\,\gamma_{t_1 T}\right] .
\end{equation}

\noindent By \eqref{eq: identity information}, we have

\begin{align}
&\mathbb{E}\left[\phi(\xi_{t},\gamma_{tT})\,\middle|\,\xi_{t_1},\,\xi_{t_2},\,\gamma_{t_1 T},\,\gamma_{t_2,T}\right]
\nonumber \\ 
&\quad\quad=\mathbb{E}\left[\phi(\xi_{t},\gamma_{tT})\,\middle|\,\xi_{t_1},\,\xi_{t_2},\,\gamma_{t_1 T},\,\gamma_{t_2 t_1}\right]
\nonumber\\
&\quad\quad=\mathbb{E}\left[\phi(\xi_{t},\gamma_{tT})\,\middle|\,\xi_{t_1},\,\Gamma_{t_2 t_1},\,\gamma_{t_1 T},\,\gamma_{t_2 t_1}\right]
\nonumber\\
&\quad\quad=\mathbb{E}\left[\phi(\Gamma_{tT}+{\gamma_{tT}}\,\sigma\,X_T,\gamma_{tT})\,\middle|\,\Gamma_{t_1 T}+{\gamma_{t_1 T}}\,\sigma\,X_T,\,\Gamma_{t_2 t_1},\,\gamma_{t_1 T},\,\gamma_{t_2 t_1}\right] .
\end{align}
Finally, we invoke Lemma \ref{Indep2}, and Theorem \ref{Markov_A_and_GammaBridge}  to conclude that

\begin{align}
&\mathbb{E}\left[\phi(\xi_{t},\gamma_{tT})|\xi_{t_1},\,\xi_{t_2},\,\gamma_{t_1 T},\,\gamma_{t_2,T}\right]
\nonumber\\
&\quad\quad=\mathbb{E}\left[\phi(\Gamma_{tT}+{\gamma_{tT}}\,\sigma\,X_T,\gamma_{tT})\,\middle|\,\Gamma_{t_1 T}+{\gamma_{t_1 T}}\,\sigma\,X_T,\,\gamma_{t_1 T}\right]
\nonumber\\
&\quad\quad=\mathbb{E}\left[\phi(\xi_{t},\gamma_{tT})|\xi_{t_1},\,\gamma_{t_1 T}\right] .
\end{align}
The generalization to $n>2$ is straightforward. 
\end{proof}
\section{Information Based Pricing}
\label{Asset Pricing}
\noindent Now we are in a position to consider the valuation of a financial asset in the setting just discussed. One recalls that $\mathbb{P}$ is understood to be the risk-neutral measure and that the interest rate is constant.  The payoff of the asset at time $T$ is taken to be an integrable random variable of the form $h(X_T)$ for some Borel function $h$, where $X_T$ is the information revealed at $T$. The filtration is generated jointly by the variance-gamma information process $\{\xi_t\}$ and the associated gamma bridge $\{\gamma_{tT}\}$. The value of the asset at time $t\in[0,T)$ is then given by the general expression  \eqref{asset pricing equation}, which on account of Theorem \ref{MarkovInformationProcess2} reduces in the present context to

\begin{align}
S_t=\re^{-r\,(T -t)}\,\mathbb{E}\left[h(X_T)\,|\,\xi_t,\,\gamma_{tT}\right], 
\end{align}
and our goal is to work out this expectation explicitly. 

Let us write $F_{X_T}$ for the \textit{a priori} distribution function of $X_T$. Thus $F_{X_T} :  x\in \mathbb{R} \mapsto F_{X_T}(x) \in [0,1]$ and we have 

\begin{equation}
F_{X_T}(x)=\mathbb{P}\left(X_T\leq x\right). 
\end{equation}
Occasionally, it will be typographically convenient to write $F_{X_T}^{(x)}$ in place of $F_{X_T}(x)$, and similarly for other distribution functions. To proceed, we require the following:

\begin{Lemma}
Let $X$ be a random variable with distribution $\{F_{X}(x)\}_{x\in \mathbb{R}}$ and let $Y$ be a continuous random variable with distribution $\{F_{Y}(y)\}_{y\in \mathbb{R}}$ and density $\{f_{Y}(y)\}_{y\in \mathbb{R}}$. Then for all $y\in\mathbb{R}$ for which $f_Y(y)\neq 0$ we have

\begin{equation}
F^{(x)}_{X|Y=y}=\frac{\int_{u\in(-\infty,x]}f_{Y|X=u}^{(y)}\,\diff F_X^{(u)}}{\int_{u\in(-\infty,\infty)}f_{Y|X=u}^{(y)}\,\diff F_X^{(u)}} ,
\label{conditional distribution of X}
\end{equation}
where $F^{(x)}_{X|Y=y}$ denotes the conditional distribution $\mathbb{P}\left(X\leq x\mid Y=y\right)$, and  where

\begin{equation}
  f_{Y|X=u}^{(y)}=\frac{\diff}{\diff y}\mathbb{P}\left(Y\leq y\mid X=u\right)  .
\end{equation}
\label{conditional density lemma}
\end{Lemma}
\begin{proof}
For \textit{any} two random variables  $X$ and $Y$ it holds that

\begin{align}
\mathbb{P}\left(X\leq x,\,Y\leq y\right)&=\mathbb{E}\left[\mathds{1}_{\{X\leq x\}}\,\mathds{1}_{\{Y\leq y\}}\right]
\nonumber\\
&=\mathbb{E}\left[\mathbb{E}\left[\mathds{1}_{\{X\leq x\}}\big|Y\right]\mathds{1}_{\{Y\leq y\}}\,\right]
\nonumber\\
&=\mathbb{E}\left[F^{(x)}_{X|Y}\,\mathds{1}_{\{Y\leq y\}}\,\right] .
\end{align}
Here we have used the fact that for each $x \in \mathbb R$ there exists a Borel measurable function 
$P_x: y \in \mathbb R \mapsto P_x(y) \in [0, 1]$ such that 
$\mathbb{E}\left[\mathds{1}_{\{X\leq x\}}\big|Y\right]= P_x(Y)$. Then for $y\in \mathbb R$ we define 

\begin{align}
F^{(x)}_{X|Y=y} = P_x(y)\,.
\end{align}
Hence

\begin{align}
\mathbb{P}\left(X\leq x,\,Y\leq y\right) 
=\int_{v\in(-\infty,y]}F^{(x)}_{X|Y=v}\,\diff F_Y^{(v)}\,.
\end{align}
By symmetry, we have

\begin{equation}
\mathbb{P}\left(X\leq x,\,Y\leq y\right)=\int_{u\in(-\infty,x]}F^{(y)}_{Y|X=u}\,\diff F_X^{(u)}\, ,
\end{equation}
from which it follows that we have the relation

\begin{equation}
\int_{u\in(-\infty,x]}F^{(y)}_{Y|X=u}\,\diff F_X^{(u)} =\int_{v\in(-\infty,y]}F^{(x)}_{X|Y=v}\,\diff F_Y^{(v)}\,.
\label{Bayes identity}
\end{equation}

Moving ahead, let us consider the measure $F_{X|Y=y}(\rd x)$ on $(\mathbb R, \mathcal B)$ defined for each $y \in \mathbb R$ by setting 

\begin{equation}
F_{X|Y=y}(A) = \mathbb{E}\left[\mathds{1}_{\{X\in A\}}\big|Y=y \right]
\end{equation} 
for any $A \in  \mathcal B$. Then $F_{X|Y=y}(\rd x)$ is absolutely continuous with respect to $F_X(\rd x)$. Indeed, suppose that $F_X(B) = 0$
for some $B \in  \mathcal B$. Now, 
$F_{X|Y=y}(B) = \mathbb{E}\left[\mathds{1}_{\{X\in B\}}\big|Y=y\right]$. 
But if
$\mathbb{E}\left[\mathds{1}_{\{X\in B\}}\right] = 0$, 
then 
$ \mathbb{E} \left[ \mathbb{E}\left[\mathds{1}_{\{X\in B\}}\big|Y\right] \right] = 0$, 
and hence 
$ \mathbb{E}\left[\mathds{1}_{\{X\in B\}}\big|Y\right] = 0$,
and therefore 
$ \mathbb{E}\left[\mathds{1}_{\{X\in B\}}\big|Y=y\right] = 0$. 
Thus $F_{X|Y=y}(B)$ vanishes for any
$B \in  \mathcal B$ for which $F_X(B)$ vanishes. It follows by the Radon-Nikodym theorem that 
for each $y \in \mathbb R$ there exists a density $\{g_y(x)\}_{x\in \mathbb{R}}$ such that

\begin{align}
F^{(x)}_{X|Y=y} = \int_{u\in(-\infty,x]} g_y(u)\,\diff F_X^{(u)}\, .
\label{Radon Nikodym}
\end{align}
Note that $\{g_y(x)\}$ is determined uniquely apart from its values on $F_X$-null sets.  Inserting \eqref{Radon Nikodym} into \eqref{Bayes identity} we obtain

\begin{equation}
\int_{u\in(-\infty,x]}F^{(y)}_{Y|X=u}\,\diff F_X^{(u)} =\int_{v\in(-\infty,y]}\int_{u\in(-\infty,x]} g_v(u)\,\diff F_X^{(u)}\,\diff F_Y^{(v)}\,,
\end{equation}
and thus by Fubini's theorem we have

\begin{equation}
\int_{u\in(-\infty,x]}F^{(y)}_{Y|X=u}\,\diff F_X^{(u)} = \int_{u\in(-\infty,x]} \int_{v\in(-\infty,y]} g_v(u) \, \diff F_Y^{(v)}\,\diff F_X^{(u)}\,.
\end{equation}
It follows then that $\{F^{(y)}_{Y|X=x}\}_{x\in \mathbb{R}}$ is determined uniquely apart from its values on $F_X$-null sets, and we have

\begin{equation}
F^{(y)}_{Y|X=x} =  \int_{v\in(-\infty,y]} g_v(x) \, \diff F_Y^{(v)}.
\label{conditional distribution of Y}
\end{equation}

This relation holds quite generally and is symmetrical between $X$ and $Y$. Indeed, we have not so far assumed that $Y$ is a continuous random variable.  If $Y$ is, in fact, a continuous random variable, then its distribution function is absolutely continuous and admits a density $\{f_Y^{(y)}\}_{y\in\mathbb{R}}$.
In that case,   \eqref{conditional distribution of Y} can be written in the form

\begin{equation}
F^{(y)}_{Y|X=x} =  \int_{v\in(-\infty,y]} g_v(x) \, f_Y^{(v)} \, \diff v \, ,
\end{equation}
from which it follows that for each value of $x$ the conditional distribution function
$\{F^{(y)}_{Y|X=x} \}_{y\in\mathbb{R}}$ is absolutely continuous and admits a density 
$\{f^{(y)}_{Y|X=x} \}_{y\in\mathbb{R}}$ such that

\begin{equation}
f^{(y)}_{Y|X=x} =  g_y(x) \, f_Y^{(y)} \, .
\label{equation for g}
\end{equation}
The desired result \eqref{conditional distribution of X} then follows from  \eqref{Radon Nikodym} and \eqref{equation for g} if we observe that 

\begin{align}
    f_Y^{(y)}=\int_{u\in(-\infty,\infty)} f_{Y|X=u}^{(y)}\,\diff F_X^{(u)}\,,
\end{align}
and that concludes the proof.
\end{proof}

Armed with Lemma \ref{conditional density lemma}, we are in a position to work out the conditional expectation that leads to the asset price, and we obtain the following: 

\vspace{0.2cm}
\begin{Theorem} The variance-gamma information-based price of a financial asset with payoff $h(X_T)$ at time $T$ is given for $t<T$ by

\begin{align}\label{eq: value v_t formula}
S_t=\re^{-r\,(T-t)}\,\int_{x\in \mathbb{R}}h(x)\,\frac{\re^{\left(\sigma\,\xi_t\,x-\frac{1}{2}\,\sigma^2\,x^2\,\gamma_{tT}\right)\,\left(1-\gamma_{tT}\right)^{-1}}}{\int_{y\in \mathbb{R}}\re^{\left(\sigma\,\xi_t\,y-\frac{1}{2}\,\sigma^2\,y^2\,\gamma_{tT}\right)\,\left(1-\gamma_{tT}\right)^{-1}}\diff F^{(y)}_{X_T}} \diff F^{(x)}_{X_T}\,.
\end{align}
\label{thm: value v_t formula}
\end{Theorem}

\begin{proof}
To calculate the conditional expectation of $h(X_T)$, we observe that

\begin{align}\label{double expectation asset pricing}
\mathbb{E}\left[h(X_T)\,|\,\xi_t,\,\gamma_{tT}\right]&=\mathbb{E}\bigg[\mathbb{E}\left[h(X_T)\,|\,\xi_t,\,\gamma_{tT},\,\gamma_T\right]\,\bigg|\,\xi_t,\,\gamma_{tT}\bigg] ,
\end{align}
by the tower property, where the inner expectation takes the form

\begin{align}
\mathbb{E}\left[h(X_T)\,|\,\xi_t=\xi,\,\gamma_{tT}=b,\,\gamma_T=g\right]&=\int_{x\in \mathbb{R}}h(x)\,\diff F^{(x)}_{X_T |\xi_t=\xi,\,\gamma_{tT}=b,\,\gamma_T=g}\,.
\end{align}
Here by Lemma \ref{conditional density lemma} the conditional distribution function is 

\begin{align}
F^{(x)}_{X_T|\xi_t=\xi,\,\gamma_{tT}=b,\,\gamma_T=g}
&=\frac{\int_{u\in(-\infty,\,x\,]} f^{(\xi)}_{\xi_t\,|\,X_T=u,\,\gamma_{tT}=b\,,\gamma_T=g}\,\diff F^{(u)}_{X_T\,|\,\gamma_{tT}=b\,,\gamma_T=g}}
{\int_{u\in \mathbb{R}}f^{(\xi)}_{\xi_t\,|\,X_T=u,\,\gamma_{tT}=b\,,\gamma_T=g}\,\diff F^{(u)}_{X_T\,|\,\gamma_{tT}=b\,,\gamma_T=g}}\nonumber\\
&=\frac{\int_{u\in(-\infty,\,x\,]} f^{(\xi)}_{\xi_t\,|\,X_T=u,\,\gamma_{tT}=b\,,\gamma_T=g}\,\diff F^{(u)}_{X_T}}
{\int_{u\in \mathbb{R}}f^{(\xi)}_{\xi_t\,|\,X_T=u,\,\gamma_{tT}=b\,,\gamma_T=g}\,\diff F^{(u)}_{X_T}}\nonumber\\
&=\frac{\int_{u\in(-\infty,\,x\,]} \re^{\left(\sigma\,\xi\,u-\frac{1}{2}\,\sigma^2\,u^2\,b\right)\,\left(1-b\right)^{-1}} \diff F^{(u)}_{X_T} }{\int_{\mathbb{R}} \re^{\left(\sigma\,\xi\,u-\frac{1}{2}\,\sigma^2\,u^2\,b\right)\,\left(1-b\right)^{-1}} \diff F^{(u)}_{X_T}}\,.
\end{align}
\noindent Therefore, the inner expectation in equation \eqref{double expectation asset pricing} is given by

\begin{align}
\mathbb{E}\left[h(X_T)\,|\,\xi_t,\,\gamma_{tT},\,\gamma_T\right]&=\int_{x\in \mathbb{R}}h(x)\,\frac{\re^{\left(\sigma\,\xi_t\,x-\frac{1}{2}\,\sigma^2\,x^2\,\gamma_{tT}\right)\,\left(1-\gamma_{tT}\right)^{-1}}}{\int_{y\in \mathbb{R}}\re^{\left(\sigma\,\xi_t\,y-\frac{1}{2}\,\sigma^2\,y^2\,\gamma_{tT}\right)\,\left(1-\gamma_{tT}\right)^{-1}}\diff F^{(y)}_{X_T}} \diff F^{(x)}_{X_T}\,.
\label{conditional provisional calculation}
\end{align}
But the right hand side of \eqref{conditional provisional calculation} depends only on $\xi_t$ and $\gamma_{tT}$. It follows immediately that 

\begin{align}
\mathbb{E}\left[h(X_T)\,|\,\xi_t,\,\gamma_{tT}\right]=\int_{x\in \mathbb{R}}h(x)\,\frac{\re^{\left(\sigma\,\xi_t\,x-\frac{1}{2}\,\sigma^2\,x^2\,\gamma_{tT}\right)\,\left(1-\gamma_{tT}\right)^{-1}}}{\int_{y\in \mathbb{R}}\re^{\left(\sigma\,\xi_t\,y-\frac{1}{2}\,\sigma^2\,y^2\,\gamma_{tT}\right)\,\left(1-\gamma_{tT}\right)^{-1}}\diff F^{(y)}_{X_T}} \diff F^{(x)}_{X_T}\,,
\end{align}
which translates  into equation \eqref{eq: value v_t formula}, and that concludes the proof.
\end{proof}

\section{Examples}
\label{Examples}

\noindent Going forward, we present some examples of variance-gamma information pricing for specific choices of (a) the payoff function $h:\mathbb{R}\to\mathbb{R}^+$ and (b) the distribution of the market factor $X_T$. 
In the figures, we display sample paths for the information processes and the corresponding prices. These paths are generated as follows. First, we simulate outcomes for the market factor $X_T$. Second, we simulate paths for the gamma process $\{\gamma_t\}_{t\geq 0}$ over the interval $[0,T]$ and an independent Brownian motion $\{W_t\}_{t\geq 0}$. Third, we evaluate the variance gamma process $\{W_{\gamma_t}\}_{t\geq 0}$ over the interval $[0,T]$ by subordinating the Brownian motion with the gamma process, and we evaluate the resulting gamma bridge $\{\gamma_{tT}\}_{0\leq t\leq T}$. Fourth, we use these ingredients to construct sample paths of the information processes, where these processes are given as in Definition \ref{def: information process}. Finally, we evaluate the pricing formula in equation \eqref{eq: value v_t formula} for each of the simulated paths and for each time step.

\vspace{0.2cm}

\noindent \textbf{Example 1.}  We begin with the simplest case, that of a unit-principal credit-risky bond without recovery.  We set $h(x) = x$, with $ \mathbb P(X_T = 0) = p_0$ and $\mathbb P(X_T = 1) = p_1$, where $p_0+p_1=1$. Thus,  we have 

\begin{equation}
    F_{X_T}(x) = p_0  \delta_0(x) + p_1  \delta_1(x)\,,
\end{equation}
where

\begin{equation}
    \delta_a(x) = \int_{y\in (-\infty, x]} \delta_a({\rm d}y)\,,
\end{equation}
and $\delta_a({\rm d}x)$ denotes the Dirac measure concentrated at the point $a$, and we are led to the following:

\vspace{0.2cm}
\begin{Proposition}
The variance-gamma information-based price of a unit-principal credit-risky discount bond with no recovery is given by

\begin{align}\label{bond value}
S_t=\re^{-r\,(T-t)}\,\frac{p_1 \, \re^{\left(\sigma\,\xi_t\,-\frac{1}{2}\,\sigma^2\,\gamma_{tT}\right)\,\left(1-\gamma_{tT}\right)^{-1}}}{p_0 + p_1 \, \re^{\left(\sigma\,\xi_t\,-\frac{1}{2}\,\sigma^2\,\gamma_{tT}\right)\,\left(1-\gamma_{tT}\right)^{-1}} } \,.
\end{align}

\end{Proposition}
\vspace{0.2cm}
\begin{figure}[H]
    \centering
    \includegraphics[scale=0.48]{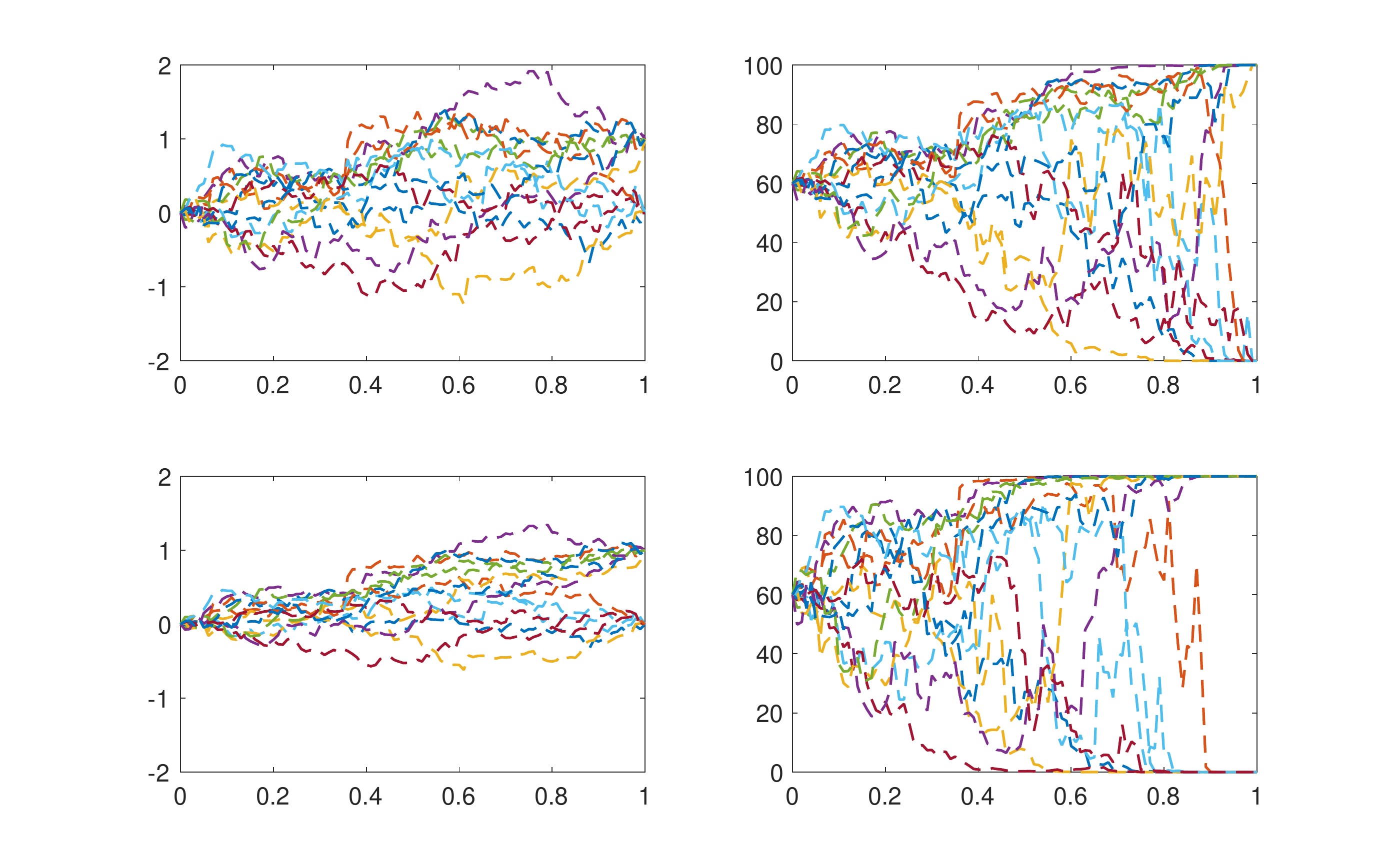}
    \caption{Credit-risky bonds with no recovery. The panels on the left show simulations of trajectories of the variance gamma information process, and the panels on the right show simulations of the corresponding price trajectories. Prices are quoted as percentages of the principal, and the interest rate is taken to be zero. From top to bottom, we show trajectories having $\sigma=1,\,2,$ respectively. We take $p_0=0.4$ for the probability of default and $p_1=0.6$ for the probability of no default. The value of $m$ is 100 in all cases.  Fifteen simulated trajectories are shown in each panel.}
    \label{fig:db 12}
\end{figure}
\noindent Now let $\omega \in \Omega$ denote the outcome of chance. By use of equation \eqref{information process} one can check rather directly that if $X_T(\omega)$ = 1, then $\lim_{t \to T} S_t = 1$, whereas if $X_T(\omega)$ = 0, then $\lim_{t \to T} S_t = 0$. More explicitly, we find that

\begin{equation}
S_t \bigg|_{X_T(w)=0}=\re^{-r(T-t)}\,\frac{p_1\,\exp\left[\sigma\,\left(\gamma^{-1/2}_T\left(W_{\gamma_t}-\gamma_{tT}\,W_{\gamma_T}\right)-\frac{1}{2}\,\sigma\,\gamma_{tT}\right)\,(1-\gamma_{tT})^{-1}\right]}{p_0+p_1\,\exp\left[\sigma\,\left(\gamma^{-1/2}_T\left(W_{\gamma_t}-\gamma_{tT}\,W_{\gamma_T}\right)-\frac{1}{2}\,\sigma\,\gamma_{tT}\right)\,(1-\gamma_{tT})^{-1}\right]}\,,
\end{equation}
whereas

\begin{equation}
S_t \bigg|_{X_T(w)=1}=\re^{-r(T-t)}\,\frac{p_1\,\exp\left[\sigma\,\left(\gamma^{-1/2}_T\left(W_{\gamma_t}-\gamma_{tT}\,W_{\gamma_T}\right)+\frac{1}{2}\,\sigma\,\gamma_{tT}\right)\,(1-\gamma_{tT})^{-1}\right]}{p_0+p_1\,\exp\left[\sigma\,\left(\gamma^{-1/2}_T\left(W_{\gamma_t}-\gamma_{tT}\,W_{\gamma_T}\right)+\frac{1}{2}\,\sigma\,\gamma_{tT}\right)\,(1-\gamma_{tT})^{-1}\right]}\,,
\end{equation}
and the claimed limiting behaviour of the asset price follows by inspection.
In Figures \ref{fig:db 12} and \ref{fig:db 34} we plot sample paths for the information processes and price processes of credit risky bonds for various values of the information flow-rate parameter. One observes that for $\sigma=1$ the information processes diverge, thus distinguishing those bonds that default from those that do not, only towards the end of the relevant time frame; whereas for higher values of $\sigma$ the divergence occurs progressively earlier, and one sees a corresponding effect in the price processes. Thus, when the information flow rate is higher, the final outcome of the bond payment is anticipated earlier, and with greater certainty. Similar conclusions hold for the interpretation of Figures \ref{fig:ln 12} and \ref{fig:ln 34}.

\begin{figure}[H]
    \centering
    \includegraphics[scale=0.48]{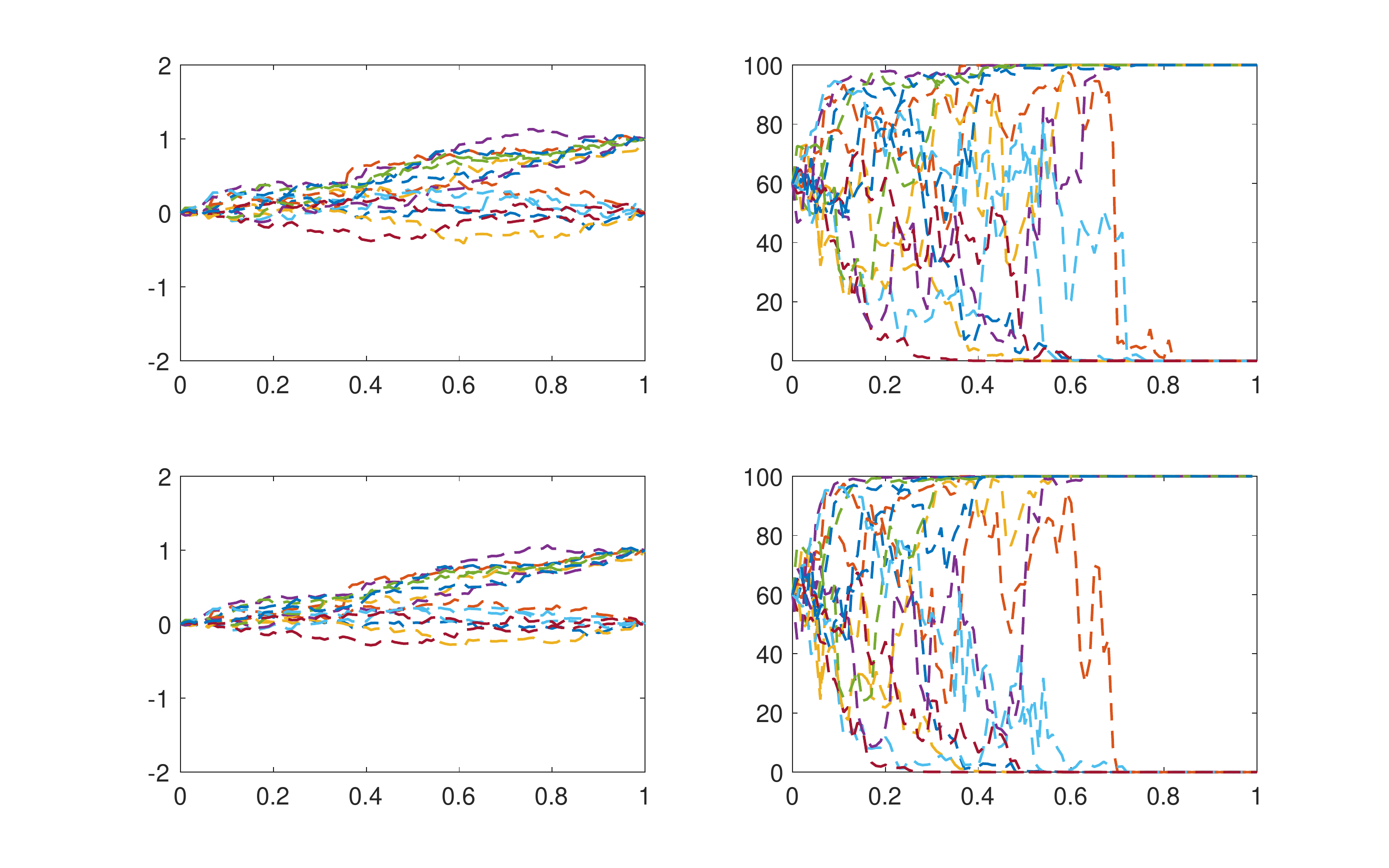}
    \caption{Credit-risky bonds with no recovery. From top to bottom we show trajectories having $\sigma=3,\,4,$ respectively. The other parameters are the same as in Figure \ref{fig:db 12}.}
    \label{fig:db 34}
\end{figure}
\vspace{0.2cm}
\noindent \textbf{Example 2.} As a somewhat more sophisticated version of the previous example, we consider the case of a defaultable bond with random recovery. We shall work out the case where $h(x)=x$ and the market factor $X_T$ takes the value $c$ with probability $p_1$ and $X_T$ is uniformly distributed over the interval $[a,b]$ with probability $p_0$, where $0\leq a<b\leq c$. Thus, for the probability measure of $X_T$ we have

\begin{equation}
F_{X_T}(\diff x)=p_0\,\mathds{1}_{\{a\leq x< b\}} \, \diff x+p_1\,\delta_c(\diff x)\,,
\end{equation}
and for the distribution function we obtain

\begin{equation}
F_{X_T}(x)=p_0\,x\,\mathds{1}_{\{a\leq x< b\}}+\mathds{1}_{\{x\geq c\}}\,.
\end{equation}
\noindent The bond price at time $t$ is then obtained by working out the expression 

\begin{equation}\label{def bond with recovery 1}
S_t=\re^{-r\,(T-t)}\,\frac{p_0 \,\int_a^b x\,\re^{\left(\sigma\,\xi_t\,x-\frac{1}{2}\,\sigma^2\,x^2\,\gamma_{tT}\right)\,\left(1-\gamma_{tT}\right)^{-1}}\,\diff x +p_1 \, c \, \re^{\left(\sigma\,\xi_t\,-\frac{1}{2}\,\sigma^2\,\gamma_{tT}\right)\,\left(1-\gamma_{tT}\right)^{-1}}}{p_0 \,\int_a^b \re^{\left(\sigma\,\xi_t\,x-\frac{1}{2}\,\sigma^2\,x^2\,\gamma_{tT}\right)\,\left(1-\gamma_{tT}\right)^{-1}}\,\diff x +p_1 \, \re^{\left(\sigma\,\xi_t\,-\frac{1}{2}\,\sigma^2\,\gamma_{tT}\right)\,\left(1-\gamma_{tT}\right)^{-1}}} \,,
\end{equation}
and it should be evident that one can obtain a closed-form solution. To work this out in detail, it will be convenient to have an expression for the incomplete first moment of a normally-distributed random variable with mean $\mu$ and variance $\nu^2$. Thus we set

\begin{equation}
N_1(x,\mu,\nu)=\frac{1}{\sqrt{2\,\pi\,\nu^2}}\,\int_{-\infty}^x y\,\exp\left({-\frac{1}{2}\,\frac{(y-\mu)^2}{\nu^2}}\right)\diff y\,,
\end{equation}
and for convenience we set 

\begin{equation}
N_0(x,\mu,\nu)=\frac{1}{\sqrt{2\,\pi\,\nu^2}}\,\int_{-\infty}^x \exp\left({-\frac{1}{2}\,\frac{(y-\mu)^2}{\nu^2}}\right)\diff y\,.
\end{equation}
Then we have 

\begin{equation}
    N_1(x,\mu,\nu)=\mu\,N\left(\frac{x-\mu}{\nu}\right)-\frac{\nu\,}{\sqrt{2\,\pi}}\,\exp\left(-\frac{1}{2}\,\frac{(x-\mu)^2}{\nu^2}\right),
\end{equation}
and of course
\begin{equation}
    N_0(x,\mu,\nu)=N\left(\frac{x-\mu}{\nu}\right),
\end{equation}
where $N(\,\cdot\,)$ is defined by \eqref{normal distribution function}.
We also set  

\begin{equation}
f(x,\mu,\nu)=\frac{1}{\sqrt{2\,\pi\,\nu^2}}\,\exp\left(-\frac{1}{2}\,\frac{(x-\mu)^2}{\nu^2}\right).
\end{equation}
Finally, we obtain the following:

\vspace{0.2cm}
\begin{Proposition}
The variance-gamma information-based price of a defaultable discount bond with a uniformly-distributed fraction of the principal paid on recovery is given by 

\begin{align}
S_t=&\re^{-r\,(T-t)}\,\frac{\,p_0 \,\big(N_1(b,{\mu},{\nu})-N_1(a,{\mu},\nu)\big) +p_1\,c\,f(c,\mu,\nu)\, }{p_0 \,\big(N_0(b,{\mu},\nu)-N_0(a,{\mu},\nu)\big) +p_1\,f(c,\mu,\nu)} ,
\end{align}
where 

\begin{equation}
{\mu}=\frac{1}{\sigma}\,\frac{\xi_t}{\gamma_{tT}}\,,\,\quad
\nu=\frac{1}{\sigma}\,\sqrt{\frac{1-\gamma_{tT}}{\gamma_{tT}}}\,.
\end{equation}
\end{Proposition}
\vspace{0.3cm}

\noindent  \textbf {Example 3}. 
Next we consider the case when the payoff  of an asset at time $T$ is log-normally distributed. This will hold if 
$h(x)=\re^{x}$ and $X_T\sim \rm{Normal}(\mu, \nu^2)$. It will be convenient to look at the slightly more general payoff obtained by setting $h(x)=\re^{q\,x}$ with $q \in \mathbb R$. If we recall the identity 

\begin{align}
\frac{1}{\sqrt{2\,\pi}}\, \int_{-\infty}^{\infty} \exp \left(-\frac{1}{2}A x^2 + Bx \right) \rd x = 
\frac{1}{\sqrt{A}}\,\exp \left(\frac{1}{2}\, \frac{\,B^2}{A}\right),
\end{align}
which holds for $A>0$ and $B \in \mathbb R$, a calculation gives 

\begin{align}\label{equation for I(q)}
I_t(q) := \int_{-\infty}^{\infty}\re^{q\,x}\,&\frac{1}{\sqrt{2\,\pi}\,\nu}\,
\exp \left[-\frac{1}{2} \frac{(x-\mu)^2}{\nu^2}\,
+ \frac{1}{1-\gamma_{tT}} \left(\sigma\,\xi_t\,x-\frac{1}{2}\,\sigma^2\,x^2\,\gamma_{tT}\right)\,\right ] \diff x \nonumber 
\\
\quad \quad \quad & = \frac{1}{\nu \sqrt{A_t}}\,\exp \left(\frac{1}{2}\, \frac{\,B_t^2}{A_t}- C \right),
\end{align}
where

\begin{align}
A_t= \frac{1- \gamma_{tT}+ \nu^2 \sigma^2\,\gamma_{tT}}{\nu^2 (1-\gamma_{tT})}\,,\quad
B_t=q+\frac{\mu}{\nu^2}+\frac{\sigma\,\xi_{t}}{1-\gamma_{tT}}\,, \quad
C= \frac{1}{2} \,\frac{\mu^2}{\nu^2}\,.
\end{align}
For $q = 1$, the price is thus given in accordance with Theorem \ref{thm: value v_t formula} by

\begin{align}
S_t = \re^{-r(T-t)}\, \frac{I_t(1)} {I_t(0)}\,.
\end{align}
Then clearly we have

\begin{equation}
    S_0=\re^{-r\,T}\,\exp \left[ \mu+\frac{1}{2}\,\nu^2 \right],
\end{equation}
and a calculation leads to the following:

\vspace{0.2cm}
\begin{Proposition}\label{Proposition 3}
The variance-gamma information-based price of a financial asset with a log-normally distributed payoff such that $\log \left(S_{T^-}\right) \sim {\rm Normal} (\mu,\nu^2)$ is given for $t\in (0,T)$ by 
\begin{align}
S_t&=\re^{r\,t}\,S_0\,\exp\left[{\frac{\nu^2\,\sigma^2\,\gamma_{tT}\,(1-\gamma_{tT})^{-1}}{1+\nu^2\,\sigma^2\,\gamma_{tT}\,(1-\gamma_{tT})^{-1}}\left(\frac{1}{\sigma\,\gamma_{tT}}\,\xi_t-\mu-\frac{1}{2}\,\nu^2\right)}\right] .
\end{align}
\end{Proposition}
\vspace{0.2cm}
\noindent More generally, one can consider the  case of a so-called power-payoff derivative for which 

\begin{align}
H_T=\left(S_{T^-}\right)^q\,,
\end{align}
where $S_{T^-}=\lim_{t\to T}S_t$ is the payoff of the asset priced above in Proposition \ref{Proposition 3}. See \cite{Bouzianis Hughston 2019} for aspects of the theory of power-payoff derivatives. In the present case if we write 

\begin{equation}
    C_t=\re^{-r\,(T-t)}\,\mathbb{E}_t\left[\left(S_{T^-}\right)^q\right]
\end{equation}
for the value of the power-payoff derivative at time $t$, we find that 

\begin{equation}
    C_t=\re^{r\,t}\,C_0\,\exp\left[ {\frac{\nu^2\,\sigma^2\,\gamma_{tT}\,(1-\gamma_{tT})^{-1}}{1+\nu^2\,\sigma^2\,\gamma_{tT}\,(1-\gamma_{tT})^{-1}}\left(\frac{q}{\sigma\,\gamma_{tT}}\,\xi_t-q\,\mu-\frac{1}{2}\,q^2\,\nu^2\right)}\right],
\end{equation}
where
 
\begin{equation}
C_0=\re^{-r\,T}\,\exp \left[ q\,\mu+\frac{1}{2}\,q^2\,\nu^2 \right].
\end{equation}
\begin{figure}[H]
    \centering
    \includegraphics[scale=0.48]{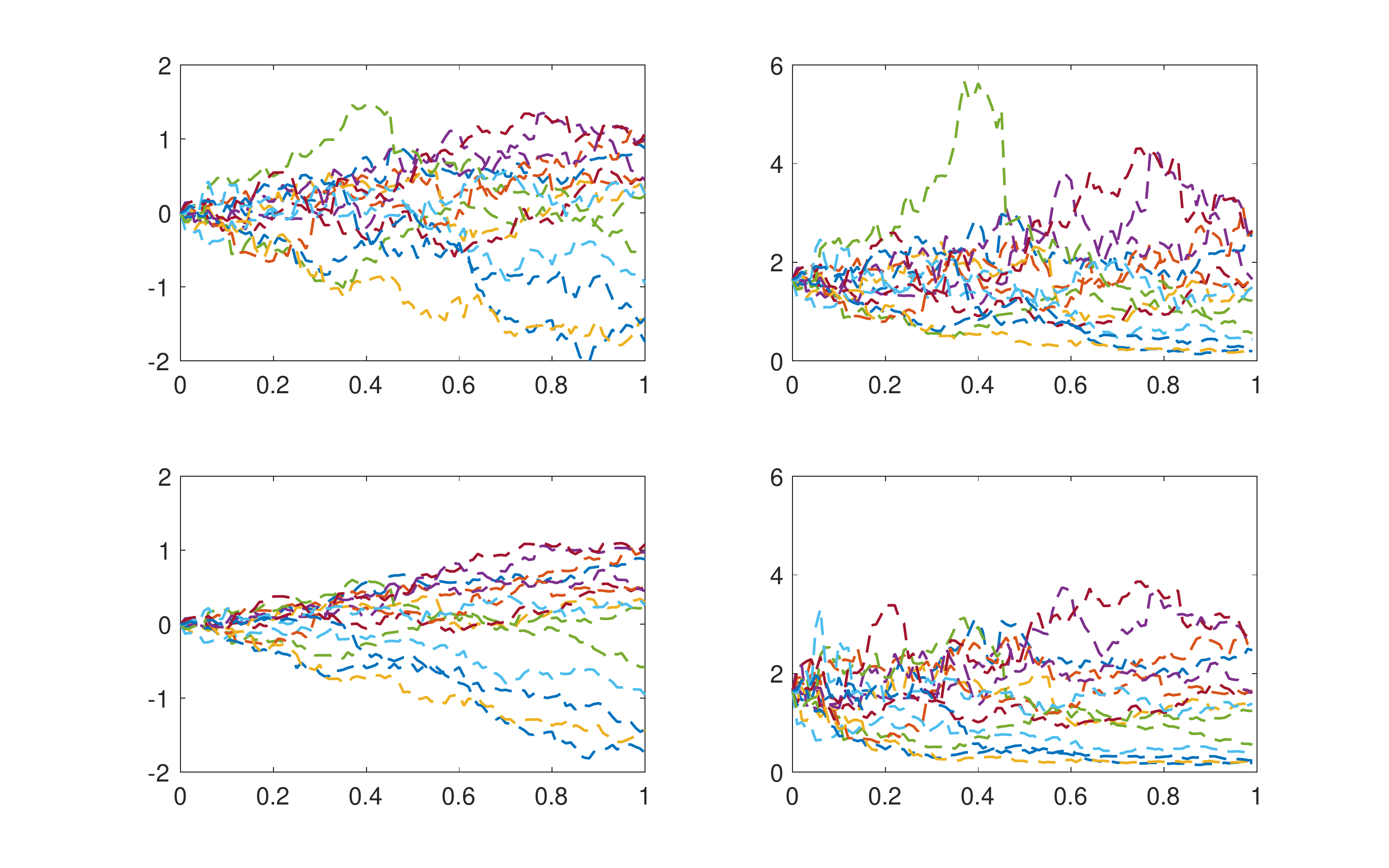}
    \caption{Log-normal payoff. The panels on the left show simulations of the trajectories of the information process, whereas the panels on the right show simulations of the corresponding price process trajectories. From the top to bottom, we show trajectories having $\sigma=1,\,2,$ respectively. The value for $m$ is 100. We take $\mu=0$, $\nu=1$, and show 15 simulated trajectories in each panel.}
    \label{fig:ln 12}
\end{figure}
\begin{figure}[H]
    \centering
    \includegraphics[scale=0.48]{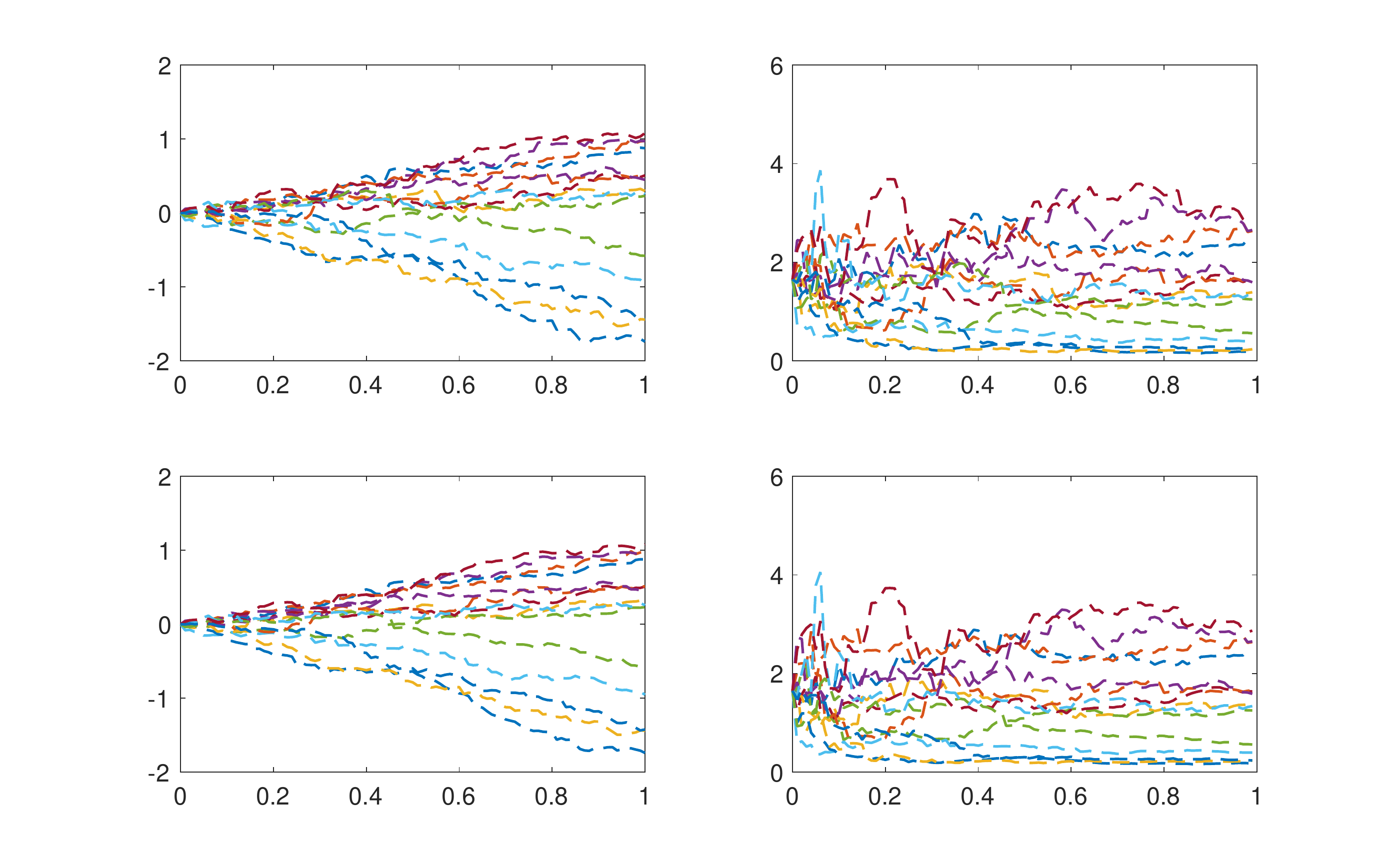}
    \caption{Log-normal payoff. From the top row to the bottom, we show trajectories having $\sigma=3,\,4,$ respectively. The other parameters are the same as those in Figure \ref{fig:ln 12}.}
    \label{fig:ln 34}
\end{figure}
%
 \noindent \textbf {Example 4.} Next we consider the case where the payoff is exponentially distributed. We let $X_T\sim \text{exp}(\lambda)$, so $\mathbb{P}\left[X_T\in\diff x\right]=\lambda\,\re^{-\lambda\,x}\,\diff x$, and take $h(x)=x$. A calculation shows that 

\begin{align}\label{eq: Close form integral 1 exp}
\int_{0}^{\infty}&{x}\,\exp{\left[-\lambda\,x+\left(\sigma\,\xi_t\,x-\frac{1}{2}\,\sigma^2\,x^2\,\gamma_{tT}\right)\,\left(1-\gamma_{tT}\right)^{-1}\right]}\,\diff x = \frac{{\mu}-N_1(0,{\mu},\nu)}{f(0,\mu,\nu)}\,,
\end{align}
where we set

\begin{align}
{\mu}=\frac{1}{\sigma}\,\frac{\xi_t}{\gamma_{tT}}-\frac{\lambda}{\sigma^2}\,\frac{1-\gamma_{tT}}{\gamma_{tT}}\,,\quad
\nu=\frac{1}{\sigma}\,\sqrt{\frac{1-\gamma_{tT}}{\gamma_{tT}}}\,,
\end{align}
and

\begin{align}\label{eq: Close form integral 2 exp}
\int_{0}^{\infty}&\exp{\left[-\lambda\,x+\left(\sigma\,\xi_t\,x-\frac{1}{2}\,\sigma^2\,x^2\,\gamma_{tT}\right)\,\left(1-\gamma_{tT}\right)^{-1}\right]}\,\diff x = \frac{1-N_0(0,{\mu},\nu)}{f(0,\mu,\nu)}\, .
\end{align}
\noindent As a consequence we obtain:

\vspace{0.2cm}
\begin{Proposition}
The variance-gamma information-based price of a financial asset with an exponentially distributed  payoff is given by
\begin{align}
S_t&=\frac{{\mu}-N_1(0,{\mu},\nu)}{1-N_0(0,{\mu},\nu)}\,,
\end{align}
where $N_0$ and $N_1$ are defined as in Example 2. 
\end{Proposition}

\section{Conclusion}
\noindent In the examples considered in the previous section, we have looked at the situation where there is a single market factor $X_T$, which is revealed at time $T$, and where the single cash flow occurring at $T$ depends on the outcome for $X_T$. The value of a security $S_t$ with that cash flow is determined by the information available at time $t$. Given the Markov property of the extended information process $\{\xi_t,\,\gamma_{tT}\}$ it follows that there exists a function of three variables $F:\mathbb{R}\times[0,1]\times\mathbb{R}^+\to \mathbb{R}^+$  such that $S_t=F(\xi_t,\,\gamma_{tT},\,t)$, and we have worked out this expression explicitly for a number of  different cases, given in Examples 1-4. The general valuation formula is presented in Theorem \ref{thm: value v_t formula}. 

It should be evident that once we have specified the functional dependence of the resulting asset prices on the extended information process, then we can back out values of the information process and the gamma bridge from the price data. So in that sense the process $\{\xi_t,\,\gamma_{tT}\}$ is 
``visible'' in the market, and can be inferred directly, at any time, from a suitable collection of prices. This means, in particular, that given the prices of a certain minimal collection of assets in the market, we can then work out the values of other assets in the market, such as derivatives. In the special case we have just been discussing, there is only a single market factor; but one can see at once that the ideas involved readily extend to the situation where there are multiple market factors and multiple cash flows, as one expects for general securities analysis, following the principles laid out in references  \cite{BHM2007, BHM2008}, where the merits and limitations of modelling in an information-based framework are discussed in some detail. 

The potential advantages of working with the variance-gamma  information process, rather than the highly tractable but more limited Brownian information process should be evident -- these include the additional parametric freedom in the model, with more flexibility in the distributions of returns, but equally important, the scope for jumps.  It comes as a pleasant surprise that the resulting formulae are to a large extent analytically explicit, but this is on account of the remarkable properties of the normalized variance-gamma bridge process that we have exploited in our constructions. Keep in mind that in the limit as the parameter $m$ goes to infinity our model reduces to that of the Brownian bridge information-based model considered in \cite{BHM2007, BHM2008}, which in turn contains the standard geometric Brownian motion model (and hence the Black-Scholes option pricing model) as a special case. In the case of a single market factor $X_T$, the  distribution of the random variable $X_T$ can be inferred by observing the current prices of derivatives for which the payoff is of the form

\begin{equation}
H_T =  \re^{rT} \mathds 1_{X_T \leq K} ,
\end{equation}
for $K \in \mathbb R$. The information flow-rate parameter $\sigma$ and the shape parameter $m$ can then be inferred from option prices. When multiple factors are involved, similar calibration methodologies are applicable.

\begin{acknowledgments}
\noindent
The authors wish to thank G.~Bouzianis and J.~M.~Pedraza-Ram\'irez for useful discussions. LSB  acknowledges  support  from (a) Oriel College, Oxford, (b) the Mathematical Institute, Oxford, (c) Consejo Nacional de Ciencia y Tenconolog\'ia (CONACyT), Ciudad de M\'exico, and (d) LMAX Exchange, London. We are grateful to the anonymous referees for a number of helpful comments and suggestions.

\end{acknowledgments}

\vspace{0.4cm}
\noindent {\bf References}



\begin{enumerate}

\bibitem{AS1970}
Abramowitz,~M. \& Stegun,~I.~A., eds.~(1972) {\em Handbook of Mathematical Functions with Formulas, Graphs, and Mathematical Tables}. United State Department of Commerce, National Bureau of Standards, Applied Mathematics Series  55. 

\bibitem{Bouzianis Hughston 2019} 
Bouzianis,~G.~\& Hughston,~L.~P.~(2019)  Determination of the L\'evy Exponent in Asset Pricing Models.
\emph{International Journal of Theoretical and Applied Finance} {\bf{22}} (1),  1950008.

\bibitem{BHM2007} 
Brody,~D.~C., Hughston,~L.~P.~\& Macrina, A.~(2007) Beyond
Hazard Rates: a New Framework for Credit-Risk Modelling. In {\it
Advances in Mathematical Finance}
(M.~C.~Fu, R.~A.~Jarrow, J.-Y.~J. Yen \& R.~J.~Elliot, eds.)  Basel: Birkh\"auser.

\bibitem{BHM2008} 
Brody,~D.~C.,~Hughston,~L.~P.~\& Macrina,~A.~(2008a)  Information-Based Asset Pricing.
\emph{International Journal of Theoretical and Applied Finance} {\bf{11}} (1), \penalty0 107-142{\natexlab{b}}.

\bibitem{BHM2008dam} Brody,~D.~C.,~Hughston,~L.~P.~\& Macrina,~A.~(2008b)
Dam Rain and Cumulative Gain.~{\em Proceeding of the Royal Society} A {\bf 464},
1801-1822.

\bibitem{BDF2009} 
Brody,~D.~C., Davis,~M.~H.~A., Friedman,~R.~L.~\& Hughston,~L.~P.~(2009) Informed Traders. {\em Proceedings of the Royal Society} A \textbf{465},
1103-1122.

\bibitem{BHM2010} 
Brody,~D.~C.,~Hughston,~L.~P.~\& Macrina,~A.~(2010) Credit Risk, Market
Sentiment and Randomly-Timed Default. In {\em Stochastic Analysis in 2010} (D. Crisan, ed.)
Berlin: Springer-Verlag.

\bibitem{BHM2011} 
Brody,~D.~C.,~Hughston,~L.~P.~\& Macrina,~A.~(2011) Modelling 
Information Flows in Financial Markets. In {\em Advanced Mathematical 
Methods for Finance} (G. Di~Nunno \& B. {\O}ksendal, eds.) Berlin: Springer-Verlag.

\bibitem{Carr Geman Madan Yor 2002} 
Carr,~P.,~Geman,~H.,~Madan,~D.~B.~\& Yor,~M.~(2002)  The Fine Structure of Asset Returns: an Empirical Investigation.
\emph{Journal of Business} {\bf 75} (2), \penalty0 305-332.

\bibitem{Emery Yor 2004} \'Emery, M.~\& Yor, M.~(2004) A Parallel between Brownian Bridges and Gamma Bridges.~{\em Publications of the Research Institute for Mathematical Sciences, Kyoto University}  \textbf{40}, 669-688.

\bibitem{FHM2012} 
Filipovi\'c,~D., Hughston,~L.~P. \& Macrina,~A.~(2012) Conditional Density Models for Asset Pricing. {\em International Journal of Theoretical and Applied Finance} \textbf{15} (1), 1250002. 

\bibitem{hoyle2010} Hoyle, E.~(2010) \textit{Information-Based Models for Finance and Insurance.} {PhD Thesis, Imperial College London}.

\bibitem{HHM2012} 
 Hoyle,~E.,~Hughston,~L.~P.~\& Macrina,~A.~(2011) 
 L\'evy Random Bridges
and the Modelling of Financial Information. {\em Stochastic Processes and
their Applications} \textbf{121}, 856-884.

\bibitem{HHM2015} 
 Hoyle,~E.,~Hughston,~L.~P.~\& Macrina,~A.~(2015) Stable-1/2 Bridges and Insurance. In {\em Advances in Mathematics of Finance} (A.~Palczewski \& L.~Stettner, eds.)  Banach Center Publications \textbf{104}, 95-120. Warsaw:  Polish Academy of Sciences.
 
 \bibitem{HMM2020} 
 Hoyle,~E., Macrina,~A.~\& Meng\"ut\"urk, L.~A.~(2020) Modulated Information Flows in Financial Markets. {\em International Journal of Theoretical and Applied Finance} \textbf{23} (4), 2050026. 

\bibitem{HM2012} 
Hughston,~L.~P.~\& Macrina,~A.~(2012) Pricing Fxed-Income Securities in an Information-Based Famework. {\em Applied Mathematical Finance} \textbf{19}  (4), 361-379.

\bibitem{Karatzas Shreve} Karatzas,~I.~\& Shreve, S.~E.~(1998) {\em
Methods of Mathematical Finance}.
New York: Springer-Verlag.

\bibitem{kyprianou2014fluctuations} Kyprianou, A.~E.~(2014) {\em
Fluctuations of L\'evy Processes with Applications}, second edition.
Berlin: Springer-Verlag.

\bibitem{Macrina2006} Macrina, A.~(2006) \textit{An Information-based Framework for Asset Pricing: X-factor Theory and its Applications}. {PhD Thesis, King's College London}.

\bibitem{MS2019}
Macrina,~A. \& Sekine,~J.~(2019) Stochastic Modelling with Randomized Markov Bridges. {\em Stochastics} \textbf{19}, 1-27.

\bibitem{Madan 1990} Madan, D.~\& Seneta, E.~(1990) The Variance Gamma (VG)
Model for Share Market
Returns.~{\em Journal of Business} \textbf{63}, 511-524.

\bibitem{Madan Milne 1991} Madan, D.~\& Milne, F.~(1991) Option Pricing with VG Martingale Components. {\em Mathematical Finance} \textbf{1} (4), 39-55.

\bibitem{Madan Carr Chang 1998} Madan, D.,~Carr, P.~\& Chang, E.~C.~(1998) The Variance
Gamma Process and Option Pricing. {\em European Finance Review} {\bf 2}, 79-105.

\bibitem{Menguturk 2013} Meng\"ut\"urk, L.~A.~(2013) \textit{Information-Based Jumps, Asymmetry and Dependence in Financial Modelling.} {PhD Thesis, Imperial College London}.

\bibitem{Menguturk 2018} Meng\"ut\"urk, L.~A.~(2018) Gaussian Random Bridges and a Geometric Model for Information Equilibrium.~{\em Physica} A {\bf 494}, 465-483.

\bibitem{Rutkowski Yu} 
Rutkowski,~R. \& Yu,~N.~(2007)  An Extension of the Brody-Hughston-Macrina Approach to Modeling of Defaultable Bonds.~{\em International Journal of Theoretical and Applied Finance} {\bf{10}} (3), \penalty0 557-589{\natexlab{b}}.

\bibitem{williams1991} Williams,~D.~(1991) {\em
Probability with Martingales}. Cambridge University Press.


\bibitem{Yor 2007} 
Yor, M.~(2007)
Some Remarkable Properties of Gamma Processes. In {\it
Advances in Mathematical Finance}
(M.~C.~Fu, R.~A.~Jarrow, J.-Y.~J. Yen \& R.~J.~Elliot, eds.)  Basel: Birkh\"auser.

\end{enumerate}

\end{document}